\documentclass{amsart}
\usepackage{amsfonts}
\usepackage{amsmath,amscd}
\usepackage{amsthm}
\usepackage{amssymb}
\usepackage{latexsym}
\usepackage{color}
\usepackage[normalem]{ulem}	
\newtheorem{lem}{Lemma}[section]

\newtheorem{prop}{Proposition}[section]
\newtheorem{note}{Note}

\newtheorem{example}{Example}[section]
\newtheorem{defn}{Definition}[section]
\newtheorem{thm}{Theorem}[section]

\newtheorem{rem}{Remark}[section]

\numberwithin{equation}{section}
\newcommand{\beqa}{\begin{eqnarray}}
\newcommand{\eeqa}{\end{eqnarray}}

\newcommand{\nc}{\newcommand}
\newcommand{\rnc}{\renewcommand}

\nc{\cal}{\mathcal}

\nc{\goth}{\mathfrak}
\rnc{\bold}{\mathbf}
\renewcommand{\frak}{\mathfrak}
\renewcommand{\Bbb}{\mathbb}

\newcommand{\fpt}[7]{{}_4\phi_3\left[\begin{matrix} #1 , #2, #3, #4 \\
#5, #6, #7 \end{matrix}\,; q,q\right]}

\nc\K{\mathbb K}
\nc{\Cal}{\mathcal}
\nc{\Xp}[1]{X^+(#1)}
\nc{\Xm}[1]{X^-(#1)}
\nc{\on}{\operatorname}
\nc{\ch}{\mbox{ch}}
\nc{\Z}{{\bold Z}}
\nc{\J}{{\mathcal J}}
\nc{\C}{{\bold C}}
\nc{\Q}{{\bold Q}}

\nc{\N}{{\Bbb N}}
\nc\beq{\begin{equation}}
\nc\enq{\end{equation}}
\nc\lan{\langle}
\nc\ran{\rangle}
\nc\bsl{\backslash}
\nc\mto{\mapsto}
\nc\lra{\leftrightarrow}
\nc\hra{\hookrightarrow}
\nc\sm{\smallmatrix}
\nc\esm{\endsmallmatrix}
\nc\sub{\subset}
\nc\ti{\tilde}
\nc\nl{\newline}
\nc\fra{\frac}
\nc\und{\underline}
\nc\ov{\overline}
\nc\ot{\otimes}
\nc\bbq{\bar{\bq}_l}
\nc\bcc{\thickfracwithdelims[]\thickness0}
\nc\ad{\text{\rm ad}}
\nc\Ad{\text{\rm Ad}}
\nc\Hom{\text{\rm Hom}}
\nc\End{\text{\rm End}}
\nc\Ind{\text{\rm Ind}}
\nc\Res{\text{\rm Res}}
\nc\Ker{\text{\rm Ker}}
\rnc\Im{\text{Im}}
\nc\sgn{\text{\rm sgn}}
\nc\tr{\text{\rm tr}}
\nc\Tr{\text{\rm Tr}}
\nc\supp{\text{\rm supp}}
\nc\card{\text{\rm card}}
\nc\bst{{}^\bigstar\!}
\nc\he{\heartsuit}
\nc\clu{\clubsuit}
\nc\spa{\spadesuit}
\nc\di{\diamond}
\nc\cW{\cal W}
\nc\cG{\cal G}
\nc\al{\alpha}
\nc\bet{\beta}
\nc\ga{\gamma}
\nc\de{\delta}
\nc\ep{\epsilon}
\nc\io{\iota}
\nc\om{\omega}
\nc\si{\sigma}
\rnc\th{\theta}
\nc\ka{\kappa}
\nc\la{\lambda}
\nc\ze{\zeta}

\nc\vp{\varpi}
\nc\vt{\vartheta}
\nc\vr{\varrho}

\nc\Ga{\Gamma}
\nc\De{\Delta}
\nc\Om{\Omega}
\nc\Si{\Sigma}
\nc\Th{\Theta}
\nc\La{\Lambda}

\nc\op{\overline{p}}
\nc\oE{\overline{E}}
\nc\onabla{\overline{\nabla}}
\nc\oDelta{\overline{\Delta}}

\nc\bepsilon{\overline{\epsilon}}
\nc\bak{\overline{k}}
\nc\bgamma{\overline{\gamma}}
\nc\kb{\mathfrak b}
\nc\hQ{\widehat Q}
\nc\kN{\mathfrak N}
\nc\kM{\mathfrak M}
\nc\kX{\mathfrak X}
\nc\kA{\mathfrak A}

\nc\kn{\mathfrak n}
\nc\km{\mathfrak m}

\nc\tfQ{\widehat{\mathfrak Q}}

\nc\boa{\bold a}
\nc\bob{\bold b}
\nc\boc{\bold c}
\nc\bod{\bold d}
\nc\boe{\bold e}
\nc\bof{\bold f}
\nc\bog{\bold g}
\nc\boh{\bold h}
\nc\boi{\bold i}
\nc\boj{\bold j}
\nc\bok{\bold k}
\nc\bol{\bold l}
\nc\bom{\bold m}
\nc\bon{\bold n}
\nc\boo{\bold o}
\nc\bop{\bold p}
\nc\boq{\bold q}
\nc\bor{\bold r}
\nc\bos{\bold s}
\nc\bou{\bold u}
\nc\bov{\bold v}
\nc\bow{\bold w}
\nc\boz{\bold z}

\nc\ba{\bold A}
\nc\bb{\bold B}
\nc\bc{\bold C}
\nc\bd{\bold D}
\nc\be{\bold E}
\nc\bg{\bold G}
\nc\bh{\bold H}
\nc\bi{\bold I}
\nc\bj{\bold J}
\nc\bk{\bold K}
\nc\bl{\bold L}
\nc\bm{\bold M}
\nc\bn{\bold N}
\nc\bo{\bold O}
\nc\bp{\bold P}
\nc\bq{\bold Q}
\nc\br{\bold R}
\nc\bs{\bold S}
\nc\bt{\bold T}
\nc\bu{\bold U}
\nc\bv{\bold V}
\nc\bw{\bold W}
\nc\bz{\bold Z}
\nc\bx{\bold X}

\nc\ca{\mathcal A}
\nc\cb{\mathcal B}
\nc\cc{\mathcal C}
\nc\cd{\mathcal D}
\nc\ce{\mathcal E}
\nc\cf{\mathcal F}
\nc\cg{\mathcal G}
\rnc\ch{\mathcal H}
\nc\ci{\mathcal I}
\nc\cj{\mathcal J}
\nc\ck{\mathcal K}
\nc\cl{\mathcal L}
\nc\cm{\mathcal M}
\nc\cn{\mathcal N}
\nc\co{\mathcal O}
\nc\cp{\mathcal P}
\nc\cq{\mathcal Q}
\nc\car{\mathcal R}
\nc\cs{\mathcal S}
\nc\ct{\mathcal T}
\nc\cu{\mathcal U}
\nc\cv{\mathcal V}
\nc\cz{\mathcal Z}
\nc\cx{\mathcal X}
\nc\cy{\mathcal Y}

\nc\e[1]{E_{#1}}
\nc\ei[1]{E_{\delta - \alpha_{#1}}}
\nc\esi[1]{E_{s \delta - \alpha_{#1}}}
\nc\eri[1]{E_{r \delta - \alpha_{#1}}}
\nc\ed[2][]{E_{#1 \delta,#2}}
\nc\ekd[1]{E_{k \delta,#1}}
\nc\emd[1]{E_{m \delta,#1}}
\nc\erd[1]{E_{r \delta,#1}}

\nc\ef[1]{F_{#1}}
\nc\efi[1]{F_{\delta - \alpha_{#1}}}
\nc\efsi[1]{F_{s \delta - \alpha_{#1}}}
\nc\efri[1]{F_{r \delta - \alpha_{#1}}}
\nc\efd[2][]{F_{#1 \delta,#2}}
\nc\efkd[1]{F_{k \delta,#1}}
\nc\efmd[1]{F_{m \delta,#1}}
\nc\efrd[1]{F_{r \delta,#1}}

\nc\fa{\frak a}
\nc\fb{\frak b}
\nc\fc{\frak c}
\nc\fd{\frak d}
\nc\fe{\frak e}
\nc\ff{\frak f}
\nc\fg{\frak g}
\nc\fh{\frak h}
\nc\fj{\frak j}
\nc\fk{\frak k}
\nc\fl{\frak l}
\nc\fm{\frak m}
\nc\fn{\frak n}
\nc\fo{\frak o}
\nc\fp{\frak p}
\nc\fq{\frak q}
\nc\fr{\frak r}
\nc\fs{\frak s}
\nc\ft{\frak t}
\nc\fu{\frak u}
\nc\fv{\frak v}
\nc\fz{\frak z}
\nc\fx{\frak x}
\nc\fy{\frak y}

\nc\fA{\frak A}
\nc\fB{\frak B}
\nc\fC{\frak C}
\nc\fD{\frak D}
\nc\fE{\frak E}
\nc\fF{\frak F}
\nc\fG{\frak G}
\nc\fH{\frak H}
\nc\fJ{\frak J}
\nc\fK{\frak K}
\nc\fL{\frak L}
\nc\fM{\frak M}
\nc\fN{\frak N}
\nc\fO{\frak O}
\nc\fP{\frak P}
\nc\fQ{\frak Q}
\nc\fR{\frak R}
\nc\fS{\frak S}
\nc\fT{\frak T}
\nc\fU{\frak U}
\nc\fV{\frak V}
\nc\fZ{\frak Z}
\nc\fX{\frak X}
\nc\fY{\frak Y}
\nc\tfi{\ti{\Phi}}
\nc\bF{\bold F}
\rnc\bol{\bold 1}

\nc\ua{\bold U_\A}

\nc\qinti[1]{[#1]_i}
\nc\q[1]{[#1]_q}
\nc\xpm[2]{E_{#2 \delta \pm \alpha_#1}}  
\nc\xmp[2]{E_{#2 \delta \mp \alpha_#1}}
\nc\xp[2]{E_{#2 \delta + \alpha_{#1}}}
\nc\xm[2]{E_{#2 \delta - \alpha_{#1}}}
\nc\hik{\ed{k}{i}}
\nc\hjl{\ed{l}{j}}
\nc\qcoeff[3]{\left[ \begin{smallmatrix} {#1}& \\ {#2}& \end{smallmatrix}
\negthickspace \right]_{#3}}
\nc\qi{q}
\nc\qj{q}

\nc\ufdm{{_\ca\bu}_{\rm fd}^{\le 0}}

\nc\isom{\cong} 

\nc{\pone}{{\Bbb C}{\Bbb P}^1}
\nc{\pa}{\partial}

\nc{\F}{{\mathcal F}}
\nc{\Sym}{{\goth S}}
\nc{\A}{{\mathcal A}}
\nc{\arr}{\rightarrow}
\nc{\larr}{\longrightarrow}

\nc{\ri}{\rangle}
\nc{\lef}{\langle}
\nc{\W}{{\mathcal W}}
\nc{\uqatwoatone}{{U_{q,1}}(\su)}
\nc{\uqtwo}{U_q(\goth{sl}_2)}
\nc{\dij}{\delta_{ij}}
\nc{\divei}{E_{\alpha_i}^{(n)}}
\nc{\divfi}{F_{\alpha_i}^{(n)}}
\nc{\Lzero}{\Lambda_0}
\nc{\Lone}{\Lambda_1}
\nc{\ve}{\varepsilon}
\nc{\phioneminusi}{\Phi^{(1-i,i)}}
\nc{\phioneminusistar}{\Phi^{* (1-i,i)}}
\nc{\phii}{\Phi^{(i,1-i)}}
\nc{\Li}{\Lambda_i}
\nc{\Loneminusi}{\Lambda_{1-i}}
\nc{\vtimesz}{v_\ve \otimes z^m}

\nc{\asltwo}{\widehat{\goth{sl}_2}}
\nc\ag{\widehat{\goth{g}}}  
\nc\teb{\tilde E_\boc}
\nc\tebp{\tilde E_{\boc'}}

\newcommand{\eeq}{\end{equation}}
\newcommand{\ben}{\begin{eqnarray}}
\newcommand{\een}{\end{eqnarray}}

\begin{document}

\title[Bispectral $q-$hypergeometric structure and $q-OA$'s integrable models]
{A bispectral  $q-$hypergeometric basis for\\   a class of quantum  integrable models}

\author{Pascal Baseilhac and Xavier Martin}
\address{Laboratoire de Math\'ematiques et Physique Th\'eorique CNRS/UMR 7350,
 F\'ed\'eration Denis Poisson FR2964,
Universit\'e de Tours,
Parc de Grammont, 37200 Tours, 
FRANCE}
\email{baseilha@lmpt.univ-tours.fr; xmartin@lmpt.univ-tours.fr}

\begin{abstract}
For the class of quantum integrable models generated from the $q-$Onsager algebra, a basis of bispectral multivariable $q-$orthogonal polynomials is exhibited. In a first part, it is shown that the multivariable Askey-Wilson polynomials with $N$ variables and $N+3$ parameters introduced by Gasper and Rahman \cite{GR1} generate a family of infinite dimensional  modules for the $q-$Onsager algebra, whose fundamental generators are realized in terms of the  multivariable $q-$difference and difference operators proposed by Iliev \cite{Iliev}. Raising and lowering operators extending those of Sahi \cite{sahi} are also constructed. In a second part, finite dimensional  modules are constructed and studied for a certain class of parameters and if the $N$ variables belong to a discrete support. In this case, the bispectral property finds  a natural interpretation within the framework of tridiagonal pairs. In a third part, eigenfunctions of the $q-$Dolan-Grady hierarchy are considered in the polynomial basis. In particular, invariant subspaces are identified for certain conditions generalizing Nepomechie's relations. In a fourth part, the analysis is extended to the special case $q=1$. This framework provides a $q-$hypergeometric formulation of quantum integrable models such as the open XXZ spin chain with generic integrable boundary conditions ($q\neq 1$).
\end{abstract}

\maketitle

\vskip -0.2cm

{\small MSC:\ 81R50;\ 81R10;\ 81U15;\ 39A70;\ 33D50;\ 39A13.}

{{\small  {\it \bf Keywords}:  $q-$Onsager algebra; Integrable systems; Bispectrality; Tridiagonal pairs; Multivariable polynomials; Basic hypergeometric series; Askey-scheme}}

\vspace{5mm}


\section{Introduction}
In the context of mathematical physics, studying the relation between $q-$hypergeometric special functions in one or several variables, related  orthogonal polynomials, and the representation theory of quantum algebras is an important field of research. From different point of views detailed in the text, this subject is challenging and deserves to be further explored for several reasons.\vspace{1mm}

Starting from the basic example of the quantum integrable harmonic oscillator whose eigenfunctions are the Hermite polynomials, generalizations involving 
a coupled system of multiple harmonic oscillators have been extensively studied in the literature (see e.g. \cite{Jeugt}).  In particular, orthogonal polynomials of the Askey scheme \cite{AW}
are eigenfunctions of Hamiltonians characterized by a certain tridiagonal interaction structure describing `triangular' two-dimensional spin lattices associated with the XY spin chain with nearest neighbor non-homogeneous couplings and non-zero magnetic field (see \cite{GMVZ} and references therein). According to the interaction considered, the parameters entering in the definition of the $q-$hypergeometric eigenfunctions are restricted. For the simplest examples of models, corresponding eigenfunctions are expressed in terms of one-variable polynomials such as the $q-$Racah, $q-$Hahn, $q-$Meixner or $q-$Krawtchouk ones. Importantly, thanks to the explicit expressions of the
 Hamiltonian's eigenfunctions in terms of $q-$hypergeometric functions and using orthogonality properties,
closed form expressions for correlation functions have been derived in various cases (see e.g. \cite{Jeugt}). In the last twenty years, such models have found applications in different areas, ranging from 
 low dimensional quantum physics and statistical mechanics to game theory, also impacting the analysis of strongly correlated systems or Markov chains. Indeed, they provide examples of transmission channels for short distance quantum  communication in which the so-called `perfect state transfer' occurs \cite{CDEL}. 
In the context of Markov chains, orthogonal polynomials also appear in the spectral representation of the transition probability matrix \cite{AGR}. Applications of Askey-Wilson polynomials in the context of solvable quantum mechanics have also been considered. See for instance \cite{OS14} and references therein.\vspace{1mm}

Multivariable extensions of the basic hypergeometric orthogonal polynomials have also arisen in the literature as eigenfunctions of mutually commuting differential or difference operators. In this vein, Macdonald-Koornwinder polynomials with $N$ variables and parameters $q,t$ \cite{M,M1,Ko,S2,Sa} are one possible generalization of the Askey-Wilson polynomials.
They provide a basis of eigenfunctions for relativisitic integrable generalizations of the Calogero-Sutherland-Moser system associated with different root systems \cite{R,D,D2,FS}. 
Also, certain degenerations of the Macdonald-Koornwinder polynomials known as $q-$Whittaker functions ($t=0$), respectively Hall-Littlewood polynomials ($q=0$), appear as eigenfunctions of the Hamiltonian of the open Toda chain with different types of boundary conditions, respectively of the repulsive delta Bose gas \cite{DE1,DE2,D3}. The limit $N\rightarrow \infty$ of Macdonald polynomials has also
been considered recently \cite{NS2}, and for $q\rightarrow 1$ Macdonald polynomials reduce to Jack polynomials which describe eigenfunctions of the quantum Benjamin-Ono model \cite{NS1}. In the context of conformal field theory and integrable perturbations, explicit formulae for singular vectors of
the Virasoro \cite{MY} and $q-$Virasoro algebra \cite{SSAFR} are given by Jack and Macdonald polynomials, respectively.  Let us  also mention that  Macdonald-Koornwinder polynomials provide a basis for a class of solutions of the boundary $q-$Knizhnik-Zamolodchikov equations \cite{SV}.\vspace{1mm} 

From the point of view of representation theory and harmonic analysis on quantum algebras, $q-$hypergeometric orthogonal polynomials can be understood 
as natural extensions of the hypergeometric functions associated with Lie groups. For $SU_q(2)$, it is known that matrix elements of irreducible representations
can be expressed in terms of little $q-$Jacobi polynomials. Askey-Wilson polynomials also found an interpretation as   `zonal spherical functions' on $SU_q(2)$ (for a review, see e.g. \cite{Koelink0}) and as overlap coefficients between eigenbasis that diagonalize certain elements of the so-called Zhedanov's algebra \cite{Zhed} (the  `Cartesian' $q-$deformation of $sl_2$, see also \cite{Fai,HPK}), also known as the Askey-Wilson algebra and
often denoted $AW(3)$ in the literature \cite{F1}. Macdonald-Koornwinder multivariable polynomials of various types
 also appear as zonal spherical functions 
in the context of the reflection equation algebra \cite{L,NDS}. 
In the representation theory of affine Hecke algebras, it is also known that each family of orthogonal polynomials can be associated with an irreducible  affine root  system \cite{C,C1}.\vspace{1mm}

Besides the Macdonald-Koornwinder generalization of the Askey-Wilson polynomials, a different type of multivariable generalization
 of the Askey-Wilson polynomials has been proposed \cite{F2} by Gasper and Rahman \cite{GR1}. These polynomials can be seen as $q-$analogs of the multivariable polynomials previously
 introduced by Tratnik \cite{Trat} (see also \cite{GerI}). In the present paper, Gasper-Rahman and Tratnik's multivariable polynomials and their connection with the $q-$Onsager algebra \cite{Ter03,Bas2} and the theory of quantum integrable systems will be at the center of our attention.\vspace{1mm}

Let us now describe the content of this paper. First, besides the two-variable case \cite{F3}, to our knowledge the relation between Gasper-Rahman multivariable generalization of
 Askey-Wilson polynomials \cite{GR1} and the representation theory of quantum algebras has not been considered up to now.
However, recall that Askey-Wilson polynomials with four parameters and $q$ appear in the representation theory of the $q-$Onsager algebra \cite{Ter03}, 
generating a simple example of infinite dimensional irreducible module. Also, observe that by restricting the four parameters and the variable to a discrete support, the Askey-Wilson polynomials generate a finite dimensional irreducible module. In this case, the 
well-known bispectral property of the Askey-Wilson polynomials finds a natural interpretation within the theory of Leonard pairs \cite{Terw0,Terw}. 
In the first part of this paper (Sections 2,3), we will fill a gap in the literature: we first show that Gasper-Rahman multivariable polynomials generate a family of infinite dimensional  
module of the $q-$Onsager algebra. We also construct raising and lowering operators which generalize those for the Askey-Wilson polynomials, proposed by \cite{sahi}. Secondly, restricting the parameters and choosing a discrete support for the $N$ variables, the bispectral property of Gasper-Rahman multivariable polynomials will be interpreted within the theory of tridiagonal pairs \cite{Ter03}. Parallel to the case $q\neq 1$, a similar analysis is handled in Section 5 for the special case $q=1$.
These are the first main results of this paper. \vspace{1mm}

Now, in the context of quantum integrable systems establishing a bridge between the representation theory of the $q-$Onsager algebra and Gasper-Rahman polynomials
opens new perspectives. Indeed, a systematic $q-$hypergeometric formulation of the eigenfunctions of all mutually commuting quantities of a given integrable model
in terms of known orthogonal polynomials of one or several variables is highly desirable. In addition to the examples mentionned at the beginning of the introduction
in which polynomials of the Askey scheme, Macdonald-Koornwinder generalizations or limiting cases appear, there are two important families of quantum integrable 
systems for which a $q-$hypergeometric formulation remains essentially, up to now, unrevealed: \vspace{1mm}

(i) The first class follows from the Faddeev-Reshetikhin-Takhtadjan and Sklyanin's (FRTS) framework, where the
Hamiltonians of quantum integrable models are derived from a generating function of mutually commuting quantities called the transfer matrix. The transfer matrix
is built from solutions of the Yang-Baxter and reflection equations which constitute the basic and main ingredients in the construction (see \cite{Skly88} for details). 
In this framework, solutions of the Yang-Baxter and reflection
equations can be interpreted as intertwiners of finite dimensional irreducible representations of quantum algebras. Heisenberg spin chains with periodic, twisted or open
 integrable boundary conditions fall among the well-known examples built in the FRTS framework, associated with $U_q(\widehat{sl_2})$ or higher rank generalizations.\vspace{1mm}

 (ii) The second class of quantum integrable systems is generated from the Onsager algebra \cite{Ons,Davies} or its $q-$deformation \cite{Bas2}. For instance, the Hamiltonians of the Ising \cite{Ons}, superintegrable
 chiral Potts \cite{Potts}, 
XY models \cite{Ar} and generalizations \cite{Ahn} are in one-to-one correspondence with the simplest element of an Abelian subalgebra of the Onsager algebra (the so-called Dolan-Grady hierarchy) acting on a finite dimensional irreducible module. Similarly, 
the Hamiltonian of the open XXZ chain with generic boundary conditions \cite{F4} can be expressed as a linear combination of mutually commuting quantities which generate the so-called $q-$Dolan-Grady hierarchy \cite{Bas2}, acting on a
finite or infinite dimensional irreducible module (in the thermodynamic or continuum limit).\vspace{1mm}

 For a large class of quantum integrable systems which have been studied in the framework (i) and/or 
(ii), it is in general not clear how to express the eigenfunctions of the transfer matrix in terms of $q-$hypergeometric functions and related polynomials. However,  in recent years it appeared that quantum integrable models
with boundaries associated with a $U_q(\widehat{sl_2})$ $R-$matrix, and $K-$matrices solutions of the reflection equation, can be alternatively generated from a $q-$deformed analog of the Onsager algebra \cite{Bas2} or, more generally, one of its higher rank generalization \cite{BB1}. In these cases, the models admit
a presentation either within the FRTS framework (i) or within the $q-$Onsager framework (ii). 
For the simplest $U_q(\widehat{sl_2})$ case, the correspondence between
the two frameworks has been studied in details \cite{BK}.\vspace{1mm}

In Section 4, we will consider a subclass of quantum integrable systems which possess two alternative presentations (i) and (ii) for generic values of $q$. The Hamiltonian $H$ - or more generally the corresponding transfer matrix - can be written as a linear combination of mutually commuting quantities $\{{\textsf I}_{2k+1}|k=0,1,2,...\}$ that form a $q-$deformed analog of the Dolan-Grady hierarchy. Namely,
\beqa
H=\sum_{k=0}^{N} h_{-k} {\textsf I}_{2k+1} + h'_0 \qquad \mbox{with} \quad [{\textsf I}_{2k+1},{\textsf I}_{2l+1}]=0, \qquad k,l \in {\mathbb N}, \label{Hamil}
\eeqa
where $h_{-k}$, $k=0,1,2,...$ and $h'_0$ are
model-dependent  scalar functions \cite{F5}.  Let
$\omega_0,\omega_1,g_+,g_-$  be scalars. The elements ${\textsf
  I}_{2k+1}$ are given by \cite{F6}: 
\beqa
{\textsf I}_{2k+1}= \omega_0{\textsf W}_{-k}+ \omega_1{\textsf W}_{k+1}+ g_+{\textsf G}_{k+1} + g_-{\tilde{\textsf G}}_{k+1}, \label{qDGhier}
\eeqa
where ${\textsf W}_{-k},{\textsf W}_{k+1},{\textsf G}_{k+1},{\tilde{\textsf G}}_{k+1}$ generate the current algebra of the $q-$ Onsager algebra \cite{BSh}. For all known examples of the literature considered up to now, ${\textsf W}_{-k},{\textsf W}_{k+1}$ (resp. ${\textsf G}_{k+1},{\tilde{\textsf G}}_{k+1}$)  are polynomials of total degree $2k+1$ (resp. $2k+2$) in the fundamental generators  ${\textsf W}_0,{\textsf W}_1$ of the $q-$Onsager algebra with defining relations shown later in (\ref{qDG}). Their explicit expressions can be computed recursively using the algorithm given in \cite{BB2}. For instance,
the elements ${\textsf{G}}_{1},{\textsf{W}}_{-1},{\textsf{G}}_{2}$ read:
\beqa
&&{\textsf{G}}_{1} = \big[{\textsf{W}_1},\textsf{W}_0\big]_q \ ,\label{defel}\\
&&{\textsf{W}}_{-1} = \frac{1}{\rho}\left( (q^2+q^{-2})\textsf{W}_0\textsf{W}_1\textsf{W}_0 -\textsf{W}_0^2\textsf{W}_1 - \textsf{W}_1 \textsf{W}_0^2\right) + \textsf{W}_1 \ ,\nonumber\\
&&{\textsf{G}}_{2} = \frac{1}{\rho(q^2+q^{-2})} \Big( (q^{-3}+q^{-1}) \textsf{W}_0^2{\textsf{W}_1}^2 - (q^{3}+q){\textsf{W}_1}^2\textsf{W}_0^2 + (q^{-3}-q^{3})(\textsf{W}_0{\textsf{W}_1}^2\textsf{W}_0 + {\textsf{W}_1}\textsf{W}_0^2{\textsf{W}_1})  \nonumber\\
&&\qquad \qquad  - (q^{-5}+q^{-3} +2q^{-1}) \textsf{W}_0{\textsf{W}_1}\textsf{W}_0{\textsf{W}_1} + (q^{5}+q^{3} +2q){\textsf{W}_1}\textsf{W}_0{\textsf{W}_1}\textsf{W}_0 +  \rho(q-q^{-1})(\textsf{W}_0^2 + {\textsf{W}_1}^2)\Big) \ .\nonumber 
\eeqa
The expressions for the elements $\tilde{\textsf{G}}_{1},{\textsf{W}}_{2},\tilde{\textsf{G}}_{2}$ are obtained from ${\textsf{G}}_{1},{\textsf{W}}_{-1},{\textsf{G}}_{2}$
by exchanging ${\textsf W}_0\leftrightarrow{\textsf W}_1$.
Note that the vector space on which ${\textsf W}_0,{\textsf W}_1$ act is specified according to the favourite model considered. Quotients of the $q-$Onsager algebra may be considered as well (see e.g. \cite{BB3}). In this  `$q-$Onsager framework', the explicit connection between Gasper-Rahman multivariable polynomials and the representation theory of the $q-$Onsager algebra mentionned above gives a straightforward access to a $q-$hypergeometric formulation of the eigenfunctions for this class of models, which is the second main result of this paper. \vspace{1mm}

The paper is organized as follows. Below, we introduce notations that will be used thoughout this paper. In Section 2, 
the orthogonal system of multivariable polynomials of Gasper and Rahman \cite{GR1}  is recalled. Then, we introduce  $q-$difference and difference operators that essentially follow from Iliev's work \cite{Iliev}, and are diagonalized by Gasper-Rahman polynomials.  Studying the structure of the operators, new infinite dimensional  modules of the $q-$Onsager algebra 
are exhibited in terms of these polynomials. In particular, explicit realizations of the standard generators of the $q-$Onsager algebra in terms of multivariable $q-$difference and difference operators are given. Generalizing the construction of Sahi \cite{sahi} to $N$-variables,  explicit realizations of raising and lowering operators are also given. In Section 3, restricting the $(N+3)-$parameter space and the variables to a discrete support, finite dimensional  modules of the $q-$Onsager algebra are constructed and described. 
In this case, the corresponding difference-difference bispectral problem is naturally interpreted within the theory of tridiagonal pairs.
 Based on these results, the spectral problem for the $q-$Dolan-Grady hierarchy of mutually commuting quantities (\ref{qDGhier}) is reconsidered in Section 4, where the eigenfunctions are shown to admit an expansion in terms of Gasper-Rahman polynomials. In particular, subspaces that are invariant under the action of the Abelian subalgebra generated by the elements (\ref{qDGhier}) are identified, extending the analysis of \cite{BK2}. These subspaces are characterized 
by certain relations that already appeared in the literature on boundary integrable models, named `Nepomechie's relations'. 
In Section 5, the construction of a hypergeometric basis for the special case $q=1$ is handled: infinite and finite dimensional  modules of the Onsager algebra and a subclass of tridiagonal algebras, respectively, are generated by Tratnik's multivariable polynomials \cite{Trat} of Krawtchouk and Racah type, respectively. In this special case $q=1$, standard generators are expressed in terms of  Iliev-Geronimo's difference operators proposed in \cite{GerI}. Perspectives for the analysis of quantum integrable models  and potential generalizations of the present framework are briefly described in the last Section.\vspace{2mm}

\subsection{Notations}
In this paper, we denote $q$ the deformation parameter, assumed not to be a root of unity. 
We will use the standard $q$-shifted factorials \cite{GR}:
\beqa
(a;q)_n=\prod_{k=0}^{n-1}(1-aq^{k}), \quad
(a;q)_{\infty}=\prod_{k=0}^{\infty}(1-aq^{k}) \quad \mbox{and} \quad
(a_1,a_2,\dots,a_k;q)_n=\prod_{j=1}^k(a_j;q)_n.\nonumber
\eeqa
\vspace{1mm}

Define $\mathcal{P}_z={\mathbb C}[ z^{\pm 1}_1,z_2^{\pm 1},...,z_N^{\pm 1}]$ as the ring of  Laurent polynomials in the variables $z_1,z_2,...,z_N$ with complex coefficients. For each of the variables $z_j$, $j=1,2,...,N$, introduce an
involution $I_j$ of $\mathcal{P}_z$ such that $I_j(z_j)=z_j^{-1}$ and $I_j(z_k)=z_k$ for $k=1,2,...,N$ and $k\not= j$. Then, we will denote $\mathcal{P}_x= {\mathbb C} [ x_1,x_2,...,x_N ]$ the subring of  $\mathcal{P}_z$  consisting of polynomials in the variables $x_i=(z_i+z_i^{-1})/2$ with complex coefficients.

\vspace{1mm}

 Let $\{e_1,e_2,\dots,e_N\}$ be the canonical basis for ${\mathbb C}^N$.  We also define \cite{F7} the $q$-shift, forward and backward difference operators in the $j$-th coordinate acting on functions $f(z)\equiv f(z_1,z_2,\dots, z_N)$, respectively, such that:
\begin{eqnarray*}
\oE_{z_j} f(z)&=&f(z_1,z_2,\dots,q^2z_j,\dots,z_N),\\ 
\oDelta_{z_j} f(z)&=&(\oE_{z_j}-1)f(z),\\  
\onabla_{z_j} f(z)&=&(1-\oE_{z_j}^{\ -1})f(z).
\end{eqnarray*}
Similarly, $E_{n_k}$ will denote the forward shift  in the variable $n_k$. For every function $f(n)\equiv f(n_1,n_2,\dots,n_N)$ we have $E_{n_k}f(n) = f(n+e_k)$.
\vspace{1mm}

The standard multi-index notation will be also used. For instance, if $\nu=(\nu_1,\nu_2,\dots,\nu_N)\in{\mathbb Z}^N$ then 
\beqa
z^{\nu}=z_1^{\nu_1}z_2^{\nu_2}\cdots z_N^{\nu_N},\quad
\oE_{z}^{\nu}=\oE_{z_1}^{\nu_1}\oE_{z_2}^{\nu_2}\cdots \oE_{z_N}^{\nu_N},\quad
zq^{\nu}=(z_1q^{\nu_1},z_2q^{\nu_2},\dots,z_Nq^{\nu_N}).\nonumber 
\eeqa
Let $\nu = (\nu_1,\nu_2,...,\nu_N)\in \{0,\pm 1\}^N \backslash \{0\}^N$. Following \cite{Iliev}, we define $\nu_j^+=max(\nu_j,0)$ and $\nu_j^-=-min(\nu_j,0)$.  We denote $I^{\nu^-}$ the composition of the involutions $I_j$ corresponding to the non-zero coordinates of $\nu^-$. Also, we denote $|\nu|=\nu_1+\nu_2+...+\nu_N$.
\vspace{4mm}

\section{The $q-$Onsager algebra and $q-$difference operators}
Tridiagonal algebras have been introduced and studied in
\cite{Ter93,Ter01,Ter03}, where they first appeared in the context
of $P-$ and $Q-$polynomial association schemes. A tridiagonal algebra is an associative algebra with unit which consists of two generators ${\textsf A}$ and
 ${\textsf A}^*$ called the standard generators (see Section 3, Definition \ref{TDgendef}). In general, the defining relations depend on five scalars $\rho,\rho^*,\gamma,\gamma^*$ and $\beta$ \cite{Ter03}. 
In this Section, we will focus on the so-called {\it reduced} parameter sequence $\gamma=0,\gamma^*=0$, $\beta=q^2+q^{-2}$ and $\rho=\rho^*$ which exhibits all  the
interesting properties and can be extended to more general parameter sequences when $q\neq 1$. In the literature, the corresponding algebra is called the $q-$Onsager algebra, 
in view of its close relationship to the Onsager algebra \cite{Ons} and the Dolan-Grady relations \cite{DG}.
\begin{defn}[\cite{Ter03,Bas2,BB1}]
The $q-$Onsager algebra $O_q(\widehat{sl_2})$ is the associative algebra with unit and standard generators $\textsf{W}_0,\textsf{W}_1$ subject to the following relations \cite{F8}
\beqa
[\textsf{W}_0,[\textsf{W}_0,[\textsf{W}_0,\textsf{W}_1]_q]_{q^{-1}}]=\rho[\textsf{W}_0,\textsf{W}_1]\
,\qquad
[\textsf{W}_1,[\textsf{W}_1,[\textsf{W}_1,\textsf{W}_0]_q]_{q^{-1}}]=\rho[\textsf{W}_1,\textsf{W}_0]\
\label{qDG} . \eeqa
\end{defn}
\begin{rem} For $\rho=0$ the relations (\ref{qDG})
reduce to the $q-$Serre relations of $U_{q}(\widehat{sl_2})$. For $q=1$, $\rho=16$ they coincide with the Dolan-Grady relations \cite{DG}.
\end{rem}

In \cite{Ter03}, an infinite dimensional irreducible module of the $q-$Onsager algebra based on the
Askey-Wilson polynomials was constructed. On the vector space of all polynomials in the single variable $x$, 
it was shown that the elements ${\textsf W}_0,{\textsf W}_1$ act, respectively, as a multiplicative
 factor in the variable $x$  and as the Askey-Wilson second order $q-$difference operator. Alternatively, in the basis of Askey-Wilson polynomials, the elements ${\textsf W}_0,{\textsf W}_1$ act as a tridiagonal matrix and diagonal matrix, respectively.  Thus, the well-known $q-$difference equation and 
three-term recurrence relations of the Askey-Wilson polynomials can be interpreted as a bispectral  $q-$difference-difference  problem for the elements ${\textsf W}_0,{\textsf W}_1$. Note that similar objects appear in previous works \cite{GLZ,GH2,NoS} and that the connection with the $q-$Onsager algebra is emphasized in \cite{Ter03}.
\vspace{1mm}

In this Section, we construct a family of infinite dimensional  modules of the $q-$Onsager algebra on the vector space $\mathcal{P}_x$ based on the Gasper-Rahman polynomials in the multivariables $x_1,x_2,...,x_N$ with coefficients in ${\mathbb C}$, thus generalizing the example of \cite{Ter03}. As we will show, to the action of ${\textsf W}_0,{\textsf W}_1$, one can associate a bispectral system of coupled recurrence relations and multivariable $q-$difference equations that is solved by the polynomials. In particular, the case $N=2$ is described in detail. Starting from the bispectral system of equations, we then extend the analysis of Sahi \cite{sahi} that holds for $N=1$ to generic values of $N$: $N$ distinct pairs of raising/lowering operators for the Gasper-Rahman polynomials are obtained, expressed in terms of the  standard generators of the $q-$Onsager algebra. 
\vspace{1mm}

\subsection{Bispectral multivariable Askey-Wilson $q-$orthogonal polynomials}
First, let us recall the definition of the multivariable polynomials introduced by Gasper and Rahman in \cite{GR1} using the notations of \cite{Iliev}. As multivariable generalizations of the well-known Askey-Wilson polynomials, these polynomials are  obtained by applying a Gram-Schmidt process. Let $q, a, b, c, d$ denote nonzero scalars in ${\mathbb C}$. To avoid degenerate situations, we assume $q$ is not a root of unity, and that none of $ab, ac, ad, bc, bd, cd, abcd$ is an integral power of $q$. For $n = 0,1,2,...$ the Askey-Wilson polynomials are defined by \cite{AW}:
\beqa
p_n(x;a,b,c,d)=\frac{(ab,ac,ad;q)_n}{a^n} \fpt{q^{-n}}{abcdq^{n-1}}{az}{az^{-1}}{ab}{ac}{ad}\quad \mbox{with}\quad x=\frac{1}{2}(z+\frac{1}{z}). \label{awpoly}
\eeqa
By definition, they satisfy a second-order $q-$difference equation and a three-term recurrence relation of the form:
\beqa
\varphi(z)p_n(x) |_{z\rightarrow q z} + \varphi(z^{-1})p_n(x) |_{z\rightarrow q^{-1} z}   + \mu(z,z^{-1}) p_n(x) &=& (q^{-n}+abcdq^{n-1}) p_n(x),\label{biAW}\\
 b_n p_{n+1}(x) + a_n p_{n}(x) + c_n p_{n-1}(x)  &=& x p_n(x)\label{biAW2},
\eeqa
where $\varphi(z),\mu(z,z^{-1})$ are rational functions of $z,z^{-1}$ and $a_n,b_n,c_n$ are rational functions of $q^n$ and of the parameters $a,b,c,d$. Their explicit expressions can be found in \cite{AW}.\vspace{1mm}

If $a,b,c,d$ are real, or occur in complex conjugate pairs if complex, and are such that max$(|a|,|b|,|c|,|d|)<1$, the polynomials satisfy the orthogonality relation \cite{F9}:
\beqa
\frac{1}{2\pi}\int_{-1}^{1}p_n(x;a,b,c,d)p_m(x;a,b,c,d)\frac{w(z;a,b,c,d)}{\sqrt{1-x^2}}dx = \delta_{n,m}\frac{(abcdq^{n-1};q)_n (abcdq^{2n};q)_\infty}{ (q^{n+1},abq^n,acq^n,adq^n,bcq^n,bdq^n,cdq^n;q)_\infty}\nonumber
\eeqa
where
\beqa
w(z;a,b,c,d)=\frac{(z^2,z^{-2};q)_{\infty}}{(az,az^{-1},bz,bz^{-1},cz,cz^{-1},dz,dz^{-1};q)_{\infty}}.\nonumber
\eeqa

In the limit $q\rightarrow 1$ and various limiting processes on the parameters $a,b,c,d$, one can obtain from the Askey-Wilson polynomials all the classical polynomials of continuous and discrete arguments: Wilson, Racah, Hahn, Jacobi, Krawtchouk, etc \cite{KS}.

\vspace{1mm}

In \cite{Trat}, a multivariable generalization of the Racah-Wilson polynomials was constructed. From that point of view, Gasper-Rahman polynomials \cite{GR1} are $q-$analogs of the multivariable polynomials introduced by Tratnik. They constitute a multivariable generalization of the Askey-Wilson polynomials (\ref{awpoly}), defined as follows:
\begin{defn}[See \cite{Iliev,GR1}] Define ${\op}_n(x;a,b,c,d)=p_n(x;a,b,c,d)|_{q\rightarrow q^2}$ from (\ref{awpoly}). Introduce the $N+3$ nonzero  parameters $\alpha_0,\alpha_1,...,\alpha_{N+2}$ and denote $x_j=\frac{1}{2}(z_j+z_j^{-1})$. Identify $z_0\equiv \alpha_0$ and $z_{N+1}\equiv \alpha_{N+2}$. The $N-$variable generalization of  the Askey-Wilson polynomials is defined by:
\beqa
Q^{(N)}(\{n\},\{x\},\{\alpha\}) = \prod_{j=1}^{N} {\op}_{n_j}(x_j;   \alpha_j q^{2\kN_{j-1}}, \frac{\alpha_j}{\alpha_0^2}q^{2\kN_{j-1}},\frac{\alpha_{j+1}}{\alpha_j}z_{j+1},\frac{\alpha_{j+1}}{\alpha_j}z^{-1}_{j+1})\label{AWgen}
\eeqa
where $\kN_j=n_1+n_2+...+n_j$ and $\kN_0=0$.
\end{defn}

\vspace{2mm}

An analog of the orthogonality property of the Askey-Wilson polynomials can be exhibited as follows. Introduce the measure on $[-1,1]^N$:
\beqa
d\mu(x) = \frac{1}{(2\pi)^N} \frac{\prod_{j=1}^N (z_j^2,z_j^{-2};q^2)_{\infty}}{  \prod_{j=0}^N \prod_{\varepsilon_1,\varepsilon_2\in\{-1,1\}}(\alpha_{j+1} \alpha_j^{-1}z^{\varepsilon_1}_{j+1}z_j^{\varepsilon_2};q^2)_{\infty} } \prod_{j=1}^N \frac{dx_j}{\sqrt{1-x_j^2}}
\eeqa
and the inner product:
\beqa
\langle f,g \rangle = \int_{[-1,1]^N} f(x)g(x) d\mu(x).\label{inprod}
\eeqa
\vspace{2mm}

Following \cite[Theorem~2.1]{Iliev}, by induction on $N$ it is straightforward to show the following theorem:
\begin{thm}[See \cite{Iliev}, Theorem 2.1]\label{thmortho} Assume:
\beqa
&&0 < |\alpha_{N+1}| <  |\alpha_{N}| < \cdots < |\alpha_1| <\mbox{min}(1,|\alpha_0|^2), \label{condpar}\\
&& \frac{|\alpha_{N+1}|}{|\alpha_N|} < |\alpha_{N+2}| <  \frac{|\alpha_{N}|}{|\alpha_{N+1}|}.\nonumber
\eeqa
The multivariable polynomials $Q^{(N)}(\{n\},\{x\},\{\alpha\})$ are orthogonal with respect to the inner product (\ref{inprod}). Namely:
\beqa
\langle Q^{(N)}(\{n\},\{x\},\{\alpha\}),Q^{(N)}(\{m\},\{x\},\{\alpha\})   \rangle = \delta_{\{n\},\{m\}} {\overline H}_{\{n\}}\label{ortH}
\eeqa
with
\beqa
{\overline H}_{\{n\}}=  \prod_{k=1}^{N}\frac{(\frac{\alpha_{k+1}^2}{\alpha_0^2} q^{2(\kN_{k-1}+\kN_k-1)} ;q^2 )_{n_k} (\frac{\alpha_{k+1}^2}{\alpha_0^2} q^{2\kN_k};q^2 )_{\infty}}{   (q^{2(n_k+1)}, \frac{\alpha_k^2}{\alpha_0^2} q^{2(\kN_{k-1}+\kN_k)},\frac{\alpha_{k+1}^2}{\alpha_0^2}q^{2(\kN_{k-1}+\kN_k)} ;q^2)_\infty } \times \frac{1}{ \prod_{\varepsilon\in \{\pm 1\}}(\alpha_{N+1}\alpha^{\varepsilon}_{N+2}q^{2\kN_N}, \frac{\alpha_{N+1}}{\alpha_0^2}\alpha_{N+2}^\varepsilon q^{2\kN_N};q^2)_\infty}\nonumber.
\eeqa
\end{thm}

\vspace{2mm}

For further convenience, normalized multivariable polynomials can be introduced following \cite{Iliev}. 
Assume the complex variables  $\{z\}=(z_1,z_2,...,z_N)\in ({\mathbb C}^*)^N$, the complex variables \cite{F10} $\{n\}=(n_1,n_2,...,n_N)\in  ({\mathbb C} \ \mbox{mod} \ \frac{2\pi i}{log(q)}{\mathbb Z})^N$ and the parameters $\{\alpha\}=(\alpha_0,\alpha_1,...,\alpha_{N+2})\in ({\mathbb C}^*)^{N+3}$.
\begin{defn}[See \cite{Iliev}]
 The normalized $N-$variable Gasper-Rahman mutlivariable polynomials is defined by:
\beqa
\qquad \hQ^{(N)}(\{n\},\{x\},\{\alpha\}) = \frac{(\alpha_{N+1}\alpha_{N+2})^{\kN_N}}{(\alpha_{N+1}\alpha_{N+2},\frac{\alpha_{N+1}\alpha_{N+2}}{\alpha_0^2};q^2)_{\kN_N}}\frac{1}{\prod_{j=1}^{N} \alpha_j^{n_j}(\frac{\alpha_{j+1}^2}{\alpha_{j}^2};q^2)_{n_j}}\ Q^{(N)}(\{n\},\{x\},\{\alpha\}). \label{normAWgen}
\eeqa
\end{defn}

\vspace{1mm}

Compared with (\ref{AWgen}), the normalization factor is chosen as follows.
Consider the following change of variables and parameters:
\beqa
\alpha_0 &\mapsto& \tilde{\alpha}_0= \alpha_0,\qquad  \alpha_j \mapsto \tilde{\alpha}_j= \frac{\alpha_0\alpha_{N+1}\alpha_{N+2}q}{\alpha_{N+2-j}} \quad \mbox{for} \quad j=1,2,...,N+1, \quad \alpha_{N+2} \mapsto \tilde{\alpha}_{N+2}= \frac{\alpha_1}{\alpha_0q},\nonumber\\ z_j &\mapsto& \tilde{z_j}= \frac{\alpha_{N+2-j}}{\alpha_0q}q^{2\kN_{N+1-j}},\nonumber\\
q^{2n_j} &\mapsto& q^{2\tilde{n}_j}=\frac{\alpha_{N+1-j} z_{N+1-j}}{\alpha_{N+2-j}z_{N+2-j}}
 \quad \mbox{for} \quad j=1,2,...,N.\nonumber
\eeqa
The map ${\mathfrak f}:(n,z,\alpha) \mapsto (\tilde{n},\tilde{z},\tilde{\alpha})$ defines an involution on $\left({\mathbb C} \ \text{mod} \ \frac{2i\pi}{log(q)}{\mathbb Z}\right)^N \times ({\mathbb C}^*)^N \times  ({\mathbb C}^*)^{N+3}$ (see  \cite[Lemma~5.1]{Iliev}). Applying Sear's formula (cf. \cite{GR}, page 49, eq. (2.10.4)), according to the factorized structure of the Gasper-Rahman polynomials (\ref{AWgen}) in terms of Askey-Wilson polynomials it is possible to show:
\begin{thm}[See \cite{Iliev}, Theorem 5.3]\label{norm} The normalized Gasper-Rahman multivariable polynomials are invariant under the action of the involution ${\mathfrak f}$:
\beqa
\hQ^{(N)}(\{n\},\{x\},\{\alpha\}) = \hQ^{(N)}(\{\tilde{n}\},\{\tilde{x}\},\{\tilde{\alpha}\}).\nonumber
\eeqa
\end{thm}

\vspace{1mm}

By analogy with the Askey-Wilson polynomials that solve the bispectral problem (\ref{biAW})-(\ref{biAW2}), the normalized Gasper-Rahman multivariable polynomials (\ref{normAWgen}) also solve a bispectral problem associated with  $q-$difference and difference operators. This is the subject of the next subsection.
\vspace{2mm}

\subsection{Iliev's $q-$difference and difference operators}
Following \cite{Iliev}, we now introduce a commutative algebra ${\cal A}_z$ of $q-$difference operators and a commutative algebra ${\cal A}_n$ of difference operators. Let 
\beqa
{\cal D}_z={\mathbb C}(z_1,z_2,...,z_N)[\oE^{\pm 1}_{z_1},\oE^{\pm 1}_{z_2},...,\oE^{\pm 1}_{z_N}]\nonumber
\eeqa
denote the associative algebra of $q-$difference operators with rational functions of $z_1,z_2,...,z_N$ as coefficients.
The commutative subalgebra ${\cal A}_z$ of ${\cal D}_z$  is generated by $N$ algebraically independent $q-$difference operators
$\{{\mathbb D}_{\{z\}}^{*(k)}|k=1,2,...,N\}$ that are defined as follows:
\begin{defn}[See \cite{Iliev}] Let $\nu = (\nu_1,\nu_2,...,\nu_N)\in \{0,\pm 1\}^N \backslash \{0\}^N$.
  Let $\{\nu_{i_1},\nu_{i_2},...,\nu_{i_s}\}$ be the nonzero components of $\nu$ with $1\leq i_1<i_2<\cdots <i_s\leq N$ and $s\geq 1$.  For $\nu\in \{0,1\}^N \backslash \{0\}^N$, denote
\beqa
\Phi_\nu(\{z\})  &=& (1-\alpha_{i_1}z_{i_1})(1-\frac{\alpha_{i_1}z_{i_1}}{\alpha_0^2})\times \frac{\prod_{k=2}^{s}    (1-\frac{\alpha_{i_k}z_{i_k}z_{i_{k-1}}}{\alpha_{i_{k-1}}})   (1-\frac{q^2\alpha_{i_k}z_{i_k}z_{i_{k-1}}}{\alpha_{i_{k-1}}}) }{ \prod_{k=1}^{s} (1-z^2_{i_k})   (1-q^2z^2_{i_k})    } \nonumber\\
&& \qquad \quad \times \ (1-\frac{\alpha_{N+1}\alpha_{N+2}z_{i_s}}{\alpha_{i_s}}) (1- \frac{\alpha_{N+1}z_{i_s}}{\alpha_{i_s}\alpha_{N+2}})\nonumber
\eeqa
and $\Phi_\nu(\{z\})=I^{\nu_-}(\Phi_{\nu_++\nu_-}(\{z\}))$ otherwise. The $q-$difference operators  $\{{\mathbb D}_{\{z\}}^{*(k)}|k=1,2,...,N\}$   are given by:
\beqa
{\mathbb D}_{\{z\}}^{*(N)} &\equiv& {\mathbb D}_{\{z\}}^{*(N)}\big(\{z_1,z_2,...,z_{N}\};\{\alpha_0,\alpha_1,...,\alpha_{N+1},\alpha_{N+2}\}\big) \label{defqdiff}\\
&=& \frac{\alpha_0 q}{\alpha_{N+1}} \Big(    \sum_{ \nu\in \{-1,0,1\}^N \backslash \{0\}^N  } \!\!\!\!\!\!\!\!\!\!\!\!(-1)^{|\nu^-|} \Phi_\nu(\{z\}) \oDelta_z^{\nu^+} \onabla_z^{\nu^-}\  + \ 1 \ + \ \frac{\alpha^2_{N+1}}{q^2\alpha_0^2}\  \Big)\nonumber\\
\mbox{and} \qquad {\mathbb D}_{\{z\}}^{*(k)} &\equiv& {\mathbb D}_{\{z\}}^{*(k)}\big(\{z_1,z_2,...,z_{k}\};\{\alpha_0,\alpha_1,...,\alpha_{k+1},z_{k+1}\}\big)  \qquad  \mbox{for}\qquad \quad k=1,2,...,N-1.\nonumber
\eeqa
\end{defn}

\vspace{3mm}

 Importantly, these operators essentially coincide with the ones given in \cite{Iliev}, up to an overall factor and additionnal constant term chosen for further convenience. According to \cite[Proposition~4.5]{Iliev}, the $q-$difference operators ${\mathbb D}_{\{z\}}^{*(k)}$, $k=1,2,...,N$ are self-adjoint and mutually commuting:
\beqa
\big[ {\mathbb D}_{\{z\}}^{*(k)}, {\mathbb D}_{\{z\}}^{*(l)}\big]=0 \quad \mbox{for all}\quad k,l\in\{1,2,...,N\}.
\eeqa
 Below, we will sometimes use an alternative formula for the $q-$difference operators introduced above as polynomials of $\{\oE_{z_j}\}$. According to    \cite[Proposition~4.2]{Iliev}, by induction on $N$ one obtains the following result:
\begin{prop}[See \cite{Iliev}, Proposition 4.2]\label{propf2} Define $z_0=\alpha_0$ and $z_{N+1}=\alpha_{N+2}$. The $N-$variable $q-$difference operator ${\mathbb D}_{\{z\}}^{*(N)}$ can be written as:
\beqa
{\mathbb D}_{\{z\}}^{*(N)} = \frac{\alpha_0q}{\alpha_{N+1}} \left( \sum_{ \nu\in \{-1,0,1\}^N  }\overline{C}_{\nu}(\{z\}) \oE_z^\nu + \frac{4\alpha_{N+1}}{\alpha_0(q^2+1)}x_0x_{N+1}\right)\label{qdiff2}
\eeqa 
where
\beqa
\overline{C}_{\nu}(z)=\left(q^2(q^2+1)\right)^{N-|{\nu}^{+}|-|\nu^{-}|}\,
\frac{\prod_{k=0}^{N}B_k^{{\nu}_k,{\nu}_{k+1}}(z)}
{\prod_{k=1}^{N}b_k^{{\nu}_k}(z)},\label{Cnu}
\eeqa
with the convention ${\nu}_0={\nu}_{N+1}=0$. For $j=0,1,...,N$,
\beqa
B_j^{0,0}(z)&=&1+\frac{\al_{j+1}^2}{q^2\al_{j}^2}-\frac{4\al_{j+1}x_{j}x_{j+1}}{(q^2+1)\al_{j}}, \qquad \qquad \qquad \quad
B_j^{1,1}(z)=\left(1-\frac{\al_{j+1}z_{j}z_{j+1}}{\al_{j}}\right)\left(1-\frac{q^2\al_{j+1}z_{j+1}z_{j}}{\al_{j}}\right),
 \nonumber\\
B_j^{0,1}(z)&=&\left(1-\frac{\al_{j+1}z_{j}z_{j+1}}{\al_{j}}\right)\left(1-\frac{\al_{j+1}z_{j+1}}{\al_{j}z_{j}}\right),\qquad
B_j^{1,0}(z)=\left(1-\frac{\al_{j+1}z_{j}z_{j+1}}{\al_{j}}\right)\left(1-\frac{\al_{j+1}z_{j}}{\al_{j}z_{j+1}}\right),
\nonumber\\
B_j^{-1,-1}(z)&=&I_{j}(I_{j+1}(B_j^{1,1}(z))),\quad B_j^{-1,l}(z)=I_j(B_j^{1,l}(z)), \quad B_j^{k,-1}(z)=I_{j+1}(B_j^{k,1}(z)) \text{ for }k,l=0,1 ,\nonumber
\eeqa
and
\beqa
b_j^{0}(z)=(1-q^2z_{j}^2)(1-q^2z_{j}^{-2}), \quad b_j^{1}(z)=(1-z_{j}^2)(1-q^2z_{j}^2),\quad
b_j^{-1}(z)=I_j(b_j^{1}(z)).\nonumber
\eeqa
\end{prop}

\vspace{2mm}

\begin{rem}\label{remarkC} According to (\ref{defqdiff}), ${\mathbb D}_{\{z\}}^{*(N)}(1)=  \ \frac{\alpha_0 q}{\alpha_{N+1}} \ + \ \frac{\alpha_{N+1}}{\alpha_0 q}\ $. From Proposition \ref{propf2}, it follows:
\beqa
\overline{C}_{\nu}(\{z\})|_{ \nu=\{0\}^N}  + \frac{4\alpha_{N+1}}{\alpha_0(q^2+1)}x_0x_{N+1}\  =  \ 1 \ + \ \frac{\alpha_{N+1}^2}{\alpha_0^2 q^2} \  - \!\!\!\!\!\! \sum_{ \nu\in \{-1,0,1\}^N\backslash \{0\}^N  \!\!\!\!\!\!\! }\overline{C}_{\nu}(\{z\}).\nonumber
\eeqa
\end{rem}

\vspace{3mm}

In view of the invariance of  the normalized Gasper-Rahman polynomials (\ref{normAWgen}) under the action of the involution ${\mathfrak f}$ (see Theorem \ref{norm}), a `dual' family of mutually commuting difference operators can be introduced.
Let 
\beqa
{\cal D}_n={\mathbb C}(q^{n_1},q^{n_2},...,q^{n_N})[E^{\pm 1}_{n_1},E^{\pm 1}_{n_2},...,E^{\pm 1}_{n_N}]\nonumber
\eeqa
denote the associative algebra of difference operators with rational functions of $q^{n_1},q^{n_2},...,q^{n_N}$ as coefficients.
Following \cite{Iliev}, introduce the map ${\mathfrak b}$  (which extends ${\mathfrak f}$) such that:
\beqa
{\kb}(\alpha_0) &=& \alpha_0,\qquad  {\kb}(\alpha_j) = \frac{\alpha_0\alpha_{N+1}\alpha_{N+2}q}{\alpha_{N+2-j}} \quad \mbox{for} \quad j=1,2,...,N+1, \quad  {\kb}(\alpha_{N+2}) = \frac{\alpha_1}{\alpha_0q},\nonumber\\
{\kb}(z_j) &=& \frac{\alpha_{N+2-j}}{\alpha_0q}q^{2(n_1+n_2+...+n_{N+1-j})},\nonumber\\
 {\kb}(\oE_{z_j}) &=& E_{n_{N+1-j}}E^{-1}_{n_{N+2-j}}  \quad \mbox{for} \quad j=1,2,...,N \quad \mbox{with} \quad E_{n_{N+1}}=Id.\nonumber
\eeqa
 The commutative subalgebra ${\cal A}_n$ of ${\cal D}_n$  is generated by $N$ algebraically independent difference operators $\{{\mathbb D}_{\{n\}}^{(k)}|k=1,2,...,N\}$ which are defined as follows:
\begin{defn}
Let $k=1,2,...,N$. The dual mutually commuting $N$-variable difference operators are defined by: 
\beqa
{\mathbb D}_{\{n\}}^{(k)}= {\kb}({\mathbb D}_{\{z\}}^{*(k)}).\label{Dnk}
\eeqa
\end{defn}

\vspace{3mm}

By construction  \cite{Iliev}, the normalized multivariable polynomials of Gasper and Rahman (\ref{normAWgen})  form a basis of the vector space ${\cal P}_x$. Extending the well-known property of the Askey-Wilson polynomials, they also solve a family of bispectral problems associated with ${\cal A}_z$ and ${\cal A}_n$. 
\begin{thm}[See \cite{Iliev}, Theorem 5.5]\label{bispec} Let $\{\alpha\}\in  ({\mathbb C}^*)^{N+3}$ and $k=1,2,...,N$. The normalized multivariable polynomial $\hQ^{(N)}(\{n\},\{x\},\{\alpha\})$ solves the following system of $q-$difference-difference bispectral problems:
\beqa
{\mathbb D}_{\{z\}}^{*(k)}  \hQ^{(N)}(\{n\},\{x\},\{\alpha\}) &=&  \left( \frac{\alpha_{k+1}}{\alpha_0 q} q^{2\kN_k} +  \frac{\alpha_0 q}{\alpha_{k+1}} q^{-2\kN_k}\right) \hQ^{(N)}(\{n\},\{x\},\{\alpha\}), \label{qdiffgen}\\
{\mathbb D}_{\{n\}}^{(k)}  \hQ^{(N)}(\{n\},\{x\},\{\alpha\}) &=&  \left( z_{N+1-k}+z_{N+1-k}^{-1}\right) \hQ^{(N)}(\{n\},\{x\},\{\alpha\})\label{recgen}
\eeqa
where
$\kN_k=n_1+n_2+\dots + n_k$.
\end{thm}

\vspace{1mm}

For $N=1$, the above equations produce  the well-known second-order $q-$difference equation (\ref{biAW}) and three-term recurrence relations  (\ref{biAW2})  of the Askey-Wilson polynomials, respectively. From the point of view of the representation theory, in the basis of Askey-Wilson polynomials, the operator ${\mathbb D}_{\{n\}}^{(1)}$ defines a semi-infinite tridiagonal matrix. For $N$ generic, in the basis of Gasper-Rahman polynomials, the corresponding matrix - which entries are determined by ${\mathbb D}_{\{n\}}^{(N)}$ - enjoys a `block' tridiagonal form, as we now show.
\begin{lem}\label{lemma} The $N-$variable difference operator ${\mathbb D}_{\{n\}}^{(N)}$ can be written as:
\beqa
 {\mathbb D}_{\{n\}}^{(N)}&=&\!\!\!\!\!\!\!\!\!\!\!\!\! \sum_{ \{\nu_2,\nu_3,\dots,\nu_N\}\in \{-1,0,1\}^{N-1}  }\!\!\!\!\!\!\!\!\!\!\!\!\! \left( 2b_{n_1 n_2 \cdots n_N}^{[\nu_N\ \nu_{N-1}-\nu_N\ \cdots \nu_2-\nu_3\ 1-\nu_2]}      E^{\nu_{N}}_{n_1}  E^{\nu_{N-1}-\nu_N}_{n_2} E^{\nu_2-\nu_3}_{n_{N-1}}  E^{1-\nu_2}_{n_N} \right. \label{qdiff3}\\
&&\qquad\qquad\qquad + \left. 2c_{n_1 n_2 \cdots n_N}^{[\nu_N\ \nu_{N-1}-\nu_N\ \cdots \nu_2-\nu_3\ -1-\nu_2]}      E^{\nu_{N}}_{n_1} E^{\nu_{N-1}-\nu_N}_{n_2}\cdots E^{\nu_2-\nu_3}_{n_{N-1}}  E^{-1-\nu_2}_{n_N} \right.\nonumber \\
 &&\qquad\qquad\qquad  +     \left. 2 a_{n_1 n_2 \cdots n_N}^{[\nu_N\ \nu_{N-1}-\nu_N\ \cdots \nu_2-\nu_3\ -\nu_2]}      E^{\nu_{N}}_{n_1} E^{\nu_{N-1}-\nu_N}_{n_2}\cdots E^{\nu_2-\nu_3}_{n_{N-1}}  E^{-\nu_2}_{n_N}\right).\nonumber
\eeqa
\end{lem}

\begin{proof} Recall that the $q-$difference multivariable operator as defined by (\ref{defqdiff}) can be alternatively written as (\ref{qdiff2}). Now, observe
\beqa
\kb(\oE_z^\nu)=\kb(\oE_{z_1}^{\nu_1}\oE_{z_2}^{\nu_2}\cdots \oE_{z_N}^{\nu_N}) &=& (E^{\nu_1}_{n_N}\underbrace{E^{-\nu_1}_{n_{N+1}}}_{\equiv Id})( E^{\nu_2}_{n_{N-1}} E^{-\nu_2}_{n_{N}}) \cdots  (E^{\nu_{N-1}}_{n_2} E^{-\nu_{N-1}}_{n_3})( E^{\nu_{N}}_{n_1} E^{-\nu_{N}}_{n_2}) \nonumber\\
&=& E^{\nu_1-\nu_2}_{n_N}   E^{\nu_2-\nu_3}_{n_{N-1}} \cdots  E^{\nu_{N-1}-\nu_N}_{n_2}  E^{\nu_{N}}_{n_1}.\nonumber
\eeqa
Decomposing ${\mathbb D}_{\{n\}}^{(N)}$ according to the values $\nu_1=+1,0,-1$, one obtains (\ref{qdiff3}).
\end{proof}
For practical applications, the explicit expressions of the coefficients entering in the decomposition (\ref{qdiff3}) are needed. 
In the next subsection, they will be given explicitly for $N=1$ - in which case the coefficients entering in the three-term recurrence relation (\ref{biAW2}) are recovered -, as well as for $N=2$.
\begin{rem}\label{coeffn}
For $N$ generic, the coefficients entering in (\ref{qdiff3}) are given by:
\beqa
b_{n_1 n_2 \cdots n_N}^{[\nu_N\ \nu_{N-1}-\nu_N\ \cdots \nu_2-\nu_3\ 1-\nu_2]}    &=& \frac{\alpha_1}{2\alpha_{N+1}\alpha_{N+2}}  \kb\left(  \overline{C}_{\nu}(\{z\})|_{\nu_1=+1} \right),\label{bcoeff}\\
c_{n_1 n_2 \cdots n_N}^{[\nu_N\ \nu_{N-1}-\nu_N\ \cdots \nu_2-\nu_3\ -1-\nu_2]}    &=& \frac{\alpha_1}{2\alpha_{N+1}\alpha_{N+2}} \kb\left( \overline{C}_{\nu}(\{z\})|_{\nu_1=-1} \right),\label{ccoeff}\\
a_{n_1 n_2 \cdots n_N}^{[\nu_N\ \nu_{N-1}-\nu_N\ \cdots \nu_2-\nu_3\ -\nu_2]}   &=& \frac{\alpha_1}{2\alpha_{N+1}\alpha_{N+2}} \kb\left(  \overline{C}_{\nu}(\{z\})|_{\nu_1=0} \right)\label{acoeff}
\eeqa
and
\beqa
a_{n_1 n_2 \cdots n_N}^{[0 0 \cdots 0 0]}   =  \frac{\alpha_1}{2\alpha_{N+1}\alpha_{N+2}}  
+ \frac{\alpha_{N+1}\alpha_{N+2}}{2\alpha_1} \!\! &-& \!\!\!\!\!\!\!\!\!\!\!\!\!\!\!\!\! \sum_{  \nu\in \{-1,0,1\}^N \backslash \{0\}^N } \left( b_{n_1 n_2 \cdots n_N}^{[\nu_N\ \nu_{N-1}-\nu_N\ \cdots \nu_2-\nu_3\ 1-\nu_2]} + c_{n_1 n_2 \cdots n_N}^{[\nu_N\ \nu_{N-1}-\nu_N\ \cdots \nu_2-\nu_3\ -1-\nu_2]} \right.\nonumber\\
&& \qquad \quad\qquad \qquad \left.+\  a_{n_1 n_2 \cdots n_N}^{[\nu_N\ \nu_{N-1}-\nu_N\ \cdots \nu_2-\nu_3\ -1-\nu_2]}    \right).\nonumber
\eeqa
\end{rem}
Note that the expression of the last coefficient follows from Remark \ref{remarkC}.\vspace{1mm}

Now, according to Lemma \ref{lemma}, it is clear that ${\mathbb D}_{\{n\}}^{(N)}$ doesn't leave invariant the eigenspace of ${\mathbb D}_{\{z\}}^{*(N)}$ generated by $\hQ^{(N)}(\{n\},\{x\},\{\alpha\})$. Precisely, the action of  ${\mathbb D}_{\{n\}}^{(N)}$ is as follows:
\vspace{1mm}

\begin{prop}\label{blocktri} Let $p\in{I\!\!N}$ be fixed. Let $V_p$  denote the eigenspace of ${\mathbb D}_{\{z\}}^{*(N)}$   generated by the normalized multivariable polynomials $\{\hQ^{(N)}(\{n\},\{x\},\{\alpha\})| \ n_1+n_2+...+n_N=p\}$. On $V_p$, the $N-$variable difference operator acts as:
\beqa
{\mathbb D}_{\{n\}}^{(N)} V_{p}  \subseteq V_{p+1} + V_{p}+ V_{p-1}.
\eeqa
\end{prop}

\begin{proof} For $N$ generic, assume $p=n_1+n_2+...+n_N$ is fixed.  Observe:
\beqa
\qquad \quad \kb(\oE_z^\nu) \hQ^{(N)}(\{n\},\{x\},\{\alpha\})= \hQ^{(N)}(n_1 +\nu_N, n_2 + \nu_{N-1}-\nu_N , ..., n_N +\nu_1-\nu_2 , \{x\},\{\alpha\})  \in V_{p+\nu_1}.
\eeqa
Using Lemma \ref{lemma} and the definition of $V_{p}$,  the claim follows.
\end{proof}

\vspace{1mm}

Above arguments obviously extend to any operator ${\mathbb D}_{\{n\}}^{(k)}$ with $k=1,2,...,N-1$. All these operators are mutually commuting semi-infinite block tridiagonal matrices.
\vspace{2mm}

\subsection{New infinite dimensional modules for the $q-$Onsager algebra}
We now endow the vector space ${\mathbb C}[x_1,x_2,...,x_N]$ with a module structure of the $q-$Onsager algebra. A family of homomorphisms indexed by the integer $k=1,2,...,N$ can be exhibited as follows.
\begin{prop}\label{realqDG2} The map defined by 
\beqa
\textsf{W}_0  \mapsto   x_{N+1-k}, \qquad \textsf{W}_1  \mapsto \frac{1}{2}{\mathbb D}_{\{z\}}^{*(k)}, \qquad \rho \mapsto -\frac{(q^2-q^{-2})^2}{4} , \label{actWgen}
\eeqa
is an homomorphism from $O_q(\widehat{sl_2})$ to ${\cal D}_z$.
\end{prop}
\begin{proof}  First, consider the case $k=N$.  Let $\cW^{(N)}_0,\cW^{(N)}_1$ denote the following linear transformation:
\beqa
{\mathbb C}[x_1,x_2,...,x_N] \qquad &\mapsto &   \qquad {\mathbb C}[x_1,x_2,...,x_N]  \nonumber
\\
\cW^{(N)}_0:\qquad \qquad\quad \qquad f\;\; \qquad  &\mapsto & \qquad \; x_1 f,
\label{linW1}
\\
\cW^{(N)}_1:\qquad \qquad \quad \qquad f\;\; \qquad & \mapsto  &\qquad \;  \frac{1}{2}{\mathbb D}_{\{z\}}^{*(N)} f.\label{linW2}
\eeqa
\vspace{1mm}

To prove the claim, we begin with the first relation in (\ref{qDG}). Let $\Delta^*$ denote  
\beqa
\Delta^*=[\cW^{(N)}_0,[\cW^{(N)}_0,[\cW^{(N)}_0,\cW^{(N)}_1]_q]_{q^{-1}}]+\frac{(q^2-q^{-2})^2}{4}[\cW^{(N)}_0,\cW^{(N)}_1] 
\eeqa
We now show $\Delta^*=0$. According to the definition of the multivariable $q-$difference operator, recall that $\cW^{(N)}_1$ is a linear combination of $\oE^{\nu_1}_{z_1}$ with $\nu_1=0,\pm 1$.  Then, observe:
\beqa
(2x_1)^2 - (q^2+q^{-2}) 2x_1 \oE^{\pm 1}_{z_1}(2x_1) + (\oE^{\pm 1}_{z_1}(2x_1))^2 + (q^2-q^{-2})^2 =0,\label{po1}
\eeqa
which implies $\Delta^*=0$.

We now turn to the second relation in  (\ref{qDG}).  Recall $\{\hQ^{(N)}(\{n\},\{x\},\{\alpha\})\}$, $\kN_N=0,1,2,...$ form a basis of the vector space  ${\mathbb C}[x_1,x_2,...,x_N]$. With respect to this basis, according to Theorem \ref{bispec}  the operator   $\cW^{(N)}_1$ is diagonalized by $\hQ^{(N)}(\{n\},\{x\},\{\alpha\})$ with eigenvalues:
\beqa
\theta^{*(N)}_{\{n\}} =   \frac{1}{2}\left(\frac{\alpha_{N+1}}{\alpha_0 q} q^{2\kN_N} +  \frac{\alpha_0 q}{\alpha_{N+1}} q^{-2\kN_N}\right).\label{spec}
\eeqa
On the other hand, by construction, the action of  $\cW^{(N)}_0$ on $\{\hQ^{(N)}(\{n\},\{x\},\{\alpha\})\}$ produces a linear combination of polynomials $\{\hQ^{(N)}(\{m\},\{x\},\{\alpha\})\}$ with $\kM_N=m_1+...+m_N$ and $\kM_N=0,1,2,...$.  Let $M_{\{n\}\{m\}}$ denote the entries of the corresponding matrix.  According to Proposition \ref{blocktri}, it has vanishing entries
 for $|\kN_N-\kM_N|>1$ i.e.  the operator $\cW^{(N)}_0$ acts as a block tridiagonal matrix in the basis $\{\hQ^{(N)}(\{n\},\{x\},\{\alpha\})\}$. Let $\Delta$ denote the matrix representing
\beqa
\Delta=[\cW^{(N)}_1,[\cW^{(N)}_1,[\cW^{(N)}_1,\cW^{(N)}_0]_q]_{q^{-1}}]+\frac{(q^2-q^{-2})^2}{4}[\cW^{(N)}_1,\cW^{(N)}_0] 
\eeqa
with respect to the basis $\{\hQ^{(N)}(\{n\},\{x\},\{\alpha\})\}$. We now show $\Delta = 0$. To do this, we show that each entry of $\Delta$  is zero. According to the action of $\cW^{(N)}_0,\cW^{(N)}_1$ on the basis  $\{\hQ^{(N)}(\{n\},\{x\},\{\alpha\})\}$, the matrix $\Delta$ has $(\{n\},\{m\})$ entry:
\beqa
\left((\theta^{*(N)}_{\{n\}})^2 - (q^2+q^{-2})\theta^{*(N)}_{\{n\}}\theta^{*(N)}_{\{m\}} + (\theta^{*(N)}_{\{m\}})^2 +\frac{(q^2-q^{-2})^2}{4}\right)\left(\theta^{*(N)}_{\{n\}}-\theta^{*(N)}_{\{m\}}\right)
M_{\{n\}\{m\}}.\label{eq}
\eeqa
Recall  $M_{\{n\}\{m\}}=0$ for $|\kN_N-\kM_N|>1$. Observe that the first factor of the expression (\ref{eq}) is vanishing for $|\kN_N-\kM_N|=1$  in view of the eigenvalues (\ref{spec}). Obviously $\theta^{*(N)}_{\{n\}}-\theta^{*(N)}_{\{m\}}=0$ for $|\kN_N-\kM_N|=0$. It follows $\Delta=0$. 
Both relations being satisfied, the claim is proven for $k=N$. For $1\leq k<N$, using  Theorem \ref{bispec}  similar arguments  are applied from which the claim follows. 
\end{proof}

\begin{rem} The above arguments are identical to those in \cite[p. 17-18]{Ter03} for $N=1$.
\end{rem}

Let us now describe explicit examples of homomorphisms for  $N=1,2$.

\begin{example}[See also \cite{Ter03}, Section 5] For $N=1$, on the vector space ${\mathbb C}[x_1]$ the elements ${\textsf W}_0,{\textsf W}_1$ of the $q-$Onsager algebra act, respectively, as ${\cal W}_0^{(1)}= x_1$ and as the $q-$difference operator:
\beqa
{\cal W}_1^{(1)} = \frac{\alpha_0 q}{2\alpha_2} \left(   \Phi_1(z_1) \oDelta_{z_1}  -  \Phi_1(z^{-1}_1)  \onabla_{z_1}+ 1+ \frac{\alpha_2^2}{q^2\alpha_0^2}   \right),\label{w11}
\eeqa
with
\beqa
\Phi_1(z_1)=\frac{(1-\alpha_1z_1)  (1-\frac{\alpha_1}{\alpha_0^2}z_1)  (1-\frac{\alpha_2\alpha_3}{\alpha_1}z_1)   (1-\frac{\alpha_2}{\alpha_1\alpha_3}z_1)}{(1-z_1^2)(1-q^2z_1^2)}.
\nonumber
\eeqa
Note that this operator coincides with the so-called Askey-Wilson second-order $q-$difference operator.
\end{example}

\vspace{2mm}

\begin{example} On the vector space ${\mathbb C}[x_1,x_2]$ the elements ${\textsf W}_0,{\textsf W}_1$ of the $q-$Onsager algebra act, respectively, as  ${\cal W}_0^{(2)}= x_1$  and as the $2-$variable $q-$difference operator:
\beqa
\qquad\quad  {\cal W}_1^{(2)} &=& \frac{\alpha_0 q}{2\alpha_3} \Big(  \Phi_{11}(z_1,z_2) \oDelta_{z_1} \oDelta_{z_2}  +   \Phi_{11}(z^{-1}_1,z^{-1}_2)  \onabla_{z_1} \onabla_{z_2}   - \  \Phi_{11}(z_1,z^{-1}_2) \oDelta_{z_1} \onabla_{z_2}  -   \Phi_{11}(z^{-1}_1,z_2)  \onabla_{z_1} \oDelta_{z_2}   \\
&& \qquad  \  + \  \Phi_{10}(z_1,z_2) \oDelta_{z_1}   -   \Phi_{10}(z_1^{-1},z_2)  \onabla_{z_1} + \Phi_{01}(z_1,z_2)  \oDelta_{z_2}  -   \Phi_{01}(z_1,z_2^{-1})  \onabla_{z_2} + 1+ \frac{\alpha_3^2}{q^2\alpha_0^2} \  \Big)\nonumber
\eeqa
with
\beqa
\Phi_{11}(z_1,z_2)&=&\frac{  (1-\alpha_1z_1)  (1-\frac{\alpha_1}{\alpha_0^2}z_1)     (1-\frac{\alpha_2}{\alpha_1}z_1z_2)  (1-q^2\frac{\alpha_2}{\alpha_1}z_1z_2)   (1-\frac{\alpha_3\alpha_4}{\alpha_2}z_2)   (1-\frac{\alpha_3}{\alpha_4\alpha_2}z_2)  }{(1-z_1^2)(1-q^2z_1^2)(1-z_2^2)(1-q^2z_2^2)}, \nonumber\\
\Phi_{10}(z_1,z_2)&=&\frac{  (1-\alpha_1z_1)  (1-\frac{\alpha_1}{\alpha_0^2}z_1)   (1-\frac{\alpha_3\alpha_4}{\alpha_1}z_1)   (1-\frac{\alpha_3}{\alpha_4\alpha_1}z_1)  }{(1-z_1^2)(1-q^2z_1^2)}, \nonumber\\
\Phi_{01}(z_1,z_2)&=&\frac{  (1-\alpha_2z_2)  (1-\frac{\alpha_2}{\alpha_0^2}z_2)  (1-\frac{\alpha_3\alpha_4}{\alpha_2}z_2)   (1-\frac{\alpha_3}{\alpha_4\alpha_2}z_2)  }{(1-z_2^2)(1-q^2z_2^2)} .\nonumber
\eeqa
\end{example}
Note that using the definition (\ref{defqdiff}), the explicit homomorphism (\ref{actWgen})  for $N=2$ and $k=1$ is given by ${\cal W}_0^{(1)}=x_2$ and  ${\cal W}_1^{(1)}$ acts as (\ref{w11}) with $\alpha_3\rightarrow z_2$.
\vspace{2mm}

For generic values of $N$, according to Lemma \ref{lemma} the spectral problem (\ref{recgen}) gives a coupled system of three-term recurrence relations with respect to the integer $\kN_N$ (see Proposition \ref{blocktri}).  In the basis of normalized Gasper-Rahman  multivariable polynomials (\ref{normAWgen}), in view of the homomorphism (\ref{actWgen}) for $k=N$, it implies that ${\textsf W}_0$ can be written as a block tridiagonal matrix which entries are given by (\ref{bcoeff}), (\ref{ccoeff}), (\ref{acoeff}). Let us now describe some explicit examples for $N=1,2$.\vspace{1mm}

\begin{example} For $N=1$, in the basis $\{\hQ^{(1)}(n_1,x_1,\alpha_0,\alpha_1,\alpha_2,\alpha_3)|n_1=0,1,...\}$ the operator ${\textsf W}_0$ acts as a semi-infinite tridiagonal matrix denoted ${\cW}^{(1)}_0$. It coincides with the well-known operator that produces the three-term recurrence relations for the Askey-Wilson polynomials. Following the conventions above, on $\hQ^{(1)}(n_1,x_1,\{\alpha\})$,
\beqa
 {\cW}^{(1)}_0 \quad \mbox{acts as}  \qquad b^{[1]}_{n_1} E_{n_1} + c^{[-1]}_{n_1} E^{-1}_{n_1}  + a^{[0]}_{n_1} \label{AWdiff}
\eeqa 
with
\beqa
b^{[1]}_{n_1}&=&\frac{\alpha_1}{2\alpha_2\alpha_3}\frac{  (1-\alpha_2\alpha_3q^{2n_1})  (1-\frac{\alpha_2\alpha_3}{\alpha_0^2}q^{2n_1})  (1-\frac{\alpha_2^2}{\alpha_0^2}q^{2n_1-2})  (1-\frac{\alpha_2^2}{\alpha_1^2}q^{2n_1})}{(1-\frac{\alpha_2^2}{\alpha_0^2}q^{4n_1-2})(1-\frac{\alpha_2^2}{\alpha_0^2}q^{4n_1})}\nonumber,\\
c^{[-1]}_{n_1}&=&\frac{\alpha_2\alpha_3}{2\alpha_1}\frac{  (1-q^{2n_1})  (1-\frac{\alpha_2}{\alpha_3}q^{2n_1-2})  (1-\frac{\alpha_2}{\alpha_3\alpha_0^2}q^{2n_1-2})  (1-\frac{\alpha_1^2}{\alpha_0^2}q^{2n_1-2})}{(1-\frac{\alpha_2^2}{\alpha_0^2}q^{4n_1-2})(1-\frac{\alpha_2^2}{\alpha_0^2}q^{4n_1-4})},\nonumber\\
a^{[0]}_{n_1}&=&\frac{\alpha_2\alpha_3}{2\alpha_1} +  \frac{\alpha_1}{2\alpha_2\alpha_3} - b^{[1]}_{n_1} -c^{[-1]}_{n_1}. \nonumber
\eeqa
\end{example}

\begin{example}  For $N=2$, in the basis $\{\hQ^{(2)}(n_1,n_2,x_1,x_2,\alpha_0,\alpha_1,\alpha_2,\alpha_3,\alpha_4)| n_1,n_2=0,1,...\}$ the operator ${\textsf W}_0$ acts as a semi-infinite block tridiagonal matrix denoted ${\cW}^{(2)}_0$. On $\hQ^{(2)}(n_1,n_2,x_1,x_2,\{\alpha\})$, 
\beqa
 {\cW}^{(2)}_0 \quad \mbox{acts as}  \quad &&  b^{[10]}_{n_1n_2} E_{n_1} +  b^{[01]}_{n_1n_2}E_{n_2} +  b^{[-12]}_{n_1n_2}E^{-1}_{n_1}E^2_{n_2} \label{AW2} \\
&&+ \   c^{[-10]}_{n_1n_2} E^{-1}_{n_1} +  c^{[0-1]}_{n_1n_2}E^{-1}_{n_2} +  c^{[1-2]}_{n_1n_2}E_{n_1}E^{-2}_{n_2}\nonumber \\
&&+  \  a^{[1-1]}_{n_1n_2} E_{n_1}E^{-1}_{n_2} +   a^{[-11]}_{n_1n_2} E^{-1}_{n_1}E_{n_2}  +   a^{[00]}_{n_1n_2}\nonumber
\eeqa
with
\beqa
b^{[10]}_{n_1n_2}&=& 
 \frac{\alpha_1}{2\alpha_3\alpha_4} \frac{                  
(1-\alpha_3\alpha_4q^{2\kN_2})(1-\frac{\alpha_3\alpha_4}{\alpha_0^2}q^{2\kN_2})(1-\frac{\alpha_3^2}{\alpha_0^2}q^{4n_1+2n_2-2})(1-\frac{\alpha_3^2}{\alpha_0^2}q^{4n_1+2n_2})(1-\frac{\alpha_2^2}{\alpha_0^2}q^{2n_1-2})(1-\frac{\alpha_2^2}{\alpha_1^2}q^{2n_1})
}{(1-\frac{\alpha_3^2}{\alpha_0^2}q^{4\kN_2-2})(1-\frac{\alpha_3^2}{\alpha_0^2}q^{4\kN_2})(1-\frac{\alpha_2^2}{\alpha_0^2}q^{4n_1-2})(1-\frac{\alpha_2^2}{\alpha_0^2}q^{4n_1})}  ,\nonumber \\
b^{[-12]}_{n_1n_2}&=&
\frac{\alpha_1}{2\alpha_3\alpha_4} \frac{                  
(1-\alpha_3\alpha_4q^{2\kN_2})(1-\frac{\alpha_3\alpha_4}{\alpha_0^2}q^{2\kN_2})(1-\frac{\alpha_3^2}{\alpha_2^2}q^{2n_2})(1-\frac{\alpha_3^2}{\alpha_2^2}q^{2n_2+2})(1-q^{-2n_1})(1-\frac{\alpha_0^2}{\alpha_1^2}q^{2-2n_1})
}{(1-\frac{\alpha_3^2}{\alpha_0^2}q^{4\kN_2-2})(1-\frac{\alpha_3^2}{\alpha_0^2}q^{4\kN_2})(1-\frac{\alpha_0^2}{\alpha_2^2}q^{2-4n_1})(1-\frac{\alpha_0^2}{\alpha_2^2}q^{4-4n_1})}  ,\nonumber \\
b^{[01]}_{n_1n_2}&=&   
\frac{\alpha_1}{2\alpha_3\alpha_4} \frac{                  
(1-\alpha_3\alpha_4q^{2\kN_2})(1-\frac{\alpha_3\alpha_4}{\alpha_0^2}q^{2\kN_2})(1-\frac{\alpha_3^2}{\alpha_0^2}q^{2\kN_2-2})(1-\frac{\alpha_3^2}{\alpha_1^2}q^{2\kN_2})}{(1-\frac{\alpha_3^2}{\alpha_0^2}q^{4\kN_2-2})(1-\frac{\alpha_3^2}{\alpha_0^2}q^{4\kN_2})}  - b^{[10]}_{n_1n_2} - b^{[-12]}_{n_1n_2},\nonumber
\eeqa
\beqa
c^{[-10]}_{n_1n_2}&=&   \frac{\alpha_3\alpha_4}{2\alpha_1} \frac{     
(1-\frac{\alpha_3}{\alpha_0^2\alpha_4}q^{2\kN_2-2})(1-\frac{\alpha_3}{\alpha_4}q^{2\kN_2-2})     
(1-\frac{\alpha_2^2}{\alpha_0^2}q^{4n_1+2n_2-2})(1-\frac{\alpha_2^2}{\alpha_0^2}q^{4n_1+2n_2-4})(1-q^{2n_1})(1-\frac{\alpha_1^2}{\alpha_0^2}q^{2n_1-2})}
{(1-\frac{\alpha_3^2}{\alpha_0^2}q^{4\kN_2-2})(1-\frac{\alpha_3^2}{\alpha_0^2}q^{4\kN_2-4})(1-\frac{\alpha_2^2}{\alpha_0^2}q^{4n_1-2})(1-\frac{\alpha_2^2}{\alpha_0^2}q^{4n_1-4})}  ,\nonumber \\ 
c^{[1-2]}_{n_1n_2}&=&
\frac{\alpha_3\alpha_4}{2\alpha_1}  \frac{                  
(1-\frac{\alpha_3}{\alpha_0^2\alpha_4}q^{2\kN_2-2})(1-\frac{\alpha_3}{\alpha_4}q^{2\kN_2-2})(1-q^{2n_2})(1-q^{2n_2-2})(1-\frac{\alpha_0^2}{\alpha_2^2}q^{2-2n_1})(1-\frac{\alpha_1^2}{\alpha_2^2}q^{-2n_1})
}{(1-\frac{\alpha_3^2}{\alpha_0^2}q^{4\kN_2-2})(1-\frac{\alpha_3^2}{\alpha_0^2}q^{4\kN_2-4})(1-\frac{\alpha_0^2}{\alpha_2^2}q^{2-4n_1})(1-\frac{\alpha_0^2}{\alpha_2^2}q^{-4n_1})},
\nonumber\\
c^{[0-1]}_{n_1n_2}&=&  \frac{\alpha_3\alpha_4}{2\alpha_1} \frac{                  
(1-\frac{\alpha_3}{\alpha_0^2\alpha_4}q^{2\kN_2-2})(1-\frac{\alpha_3}{\alpha_4}q^{2\kN_2-2})(1-q^{2\kN_2})(1-\frac{\alpha_1^2}{\alpha_0^2}q^{2\kN_2-2})}{(1-\frac{\alpha_3^2}{\alpha_0^2}q^{4\kN_2-2})(1-\frac{\alpha_3^2}{\alpha_0^2}q^{4\kN_2-4})} - c^{[-10]}_{n_1n_2} - c^{[1-2]}_{n_1n_2},
\nonumber
\eeqa
\beqa
a^{[1-1]}_{n_1n_2}&=&  
\frac{\alpha_1}{2\alpha_3\alpha_4} \frac{                  
(1-\frac{\alpha_2^2}{\alpha_0^2}q^{2n_1-2})(1-\frac{\alpha_2^2}{\alpha_1^2}q^{2n_1})(1-\alpha_3\alpha_4q^{2n_1})(1-\frac{\alpha_3\alpha_4}{\alpha_0^2}q^{2n_1})}{(1-\frac{\alpha_2^2}{\alpha_0^2}q^{4n_1-2})(1-\frac{\alpha_2^2}{\alpha_0^2}q^{4n_1})} 
- b^{[10]}_{n_1n_2} - c^{[1-2]}_{n_1n_2}
,\nonumber \\
a^{[-11]}_{n_1n_2}&=& \frac{\alpha_3\alpha_4}{2\alpha_1} \frac{                  
(1-q^{2n_1})(1-\frac{\alpha_1^2}{\alpha_0^2}q^{2n_1-2})(1-\frac{\alpha^2_2}{\alpha_0^2\alpha_3\alpha_4}q^{2n_1-2})(1-\frac{\alpha_2^2}{\alpha_3\alpha_4}q^{2n_1-2})}{(1-\frac{\alpha_2^2}{\alpha_0^2}q^{4n_1-2})(1-\frac{\alpha_2^2}{\alpha_0^2}q^{4n_1-4})} 
- c^{[-10]}_{n_1n_2} - b^{[-12]}_{n_1n_2}
,\nonumber \\
a^{[00]}_{n_1n_2}&=& \frac{\alpha_3\alpha_4}{2\alpha_1} +  \frac{\alpha_1}{2\alpha_3\alpha_4} -a^{[1-1]}_{n_1n_2} - a^{[-11]}_{n_1n_2} - b^{[-12]}_{n_1n_2} - b^{[10]}_{n_1n_2} - b^{[01]}_{n_1n_2} - c^{[1-2]}_{n_1n_2} - c^{[-10]}_{n_1n_2} - c^{[0-1]}_{n_1n_2}.\nonumber
\eeqa
\end{example}

More generally, realizations of the elements of the $q-$Onsager algebra in terms of $q-$difference operators are obtained in a straightforward manner using Proposition \ref{realqDG2}. Also, realizations in terms of difference operators  are obtained from the following identification: on the polynomial $\hQ^{(N)}(\{n\},\{x\},\{\alpha\})$,
\beqa
\cW_0^{(N)} \quad \mbox{acts as}  \qquad \frac{1}{2} {\mathbb D}_{\{n\}}^{(N)}\qquad  \mbox{and} \qquad  \cW_1^{(N)} \quad \mbox{acts as}  \qquad\theta^{*(N)}_{\{n\}}\label{real2}
\eeqa
with (\ref{qdiff3}), the expressions given in Remark \ref{coeffn}, and (\ref{spec}). \vspace{2mm}

   Note that for certain sets of parameters $\alpha_0,\alpha_1,...,\alpha_{N+2}$, the infinite dimensional module of the $q-$Onsager algebra generated by the Gasper-Rahman polynomials may become reducible (see Section 3). Also, the conditions of irreducibility of this module are not studied in this paper.\vspace{1mm} 

Combining previous results, we obtain one of the main results of this paper:
\begin{thm}
 Let $V$ denote the infinite dimensional vector space ${\cal P}_x$.  Let $\textsf{W}_0,\textsf{W}_1$ be the standard generators of the $q-$Onsager algebra (\ref{qDG}) with $q\neq 1$, where $\textsf{W}_0,\textsf{W}_1$ act on $V$ as (\ref{linW1}), (\ref{linW2}).
Then, $\textsf{W}_0,\textsf{W}_1$ induce on $V$ a module structure for the $q-$Onsager algebra with $\rho=-(q^2-q^{-2})^2/4$.
The Gasper-Raman multivariable polynomials  $\{\hQ^{(N)}(\{n\},\{x\},\{\alpha\})|\{n\}\in ({\mathbb N})^N\}$  form a basis of this module. 
\end{thm}

In Section 3, we will consider subspaces of parameters $\alpha_0,\alpha_1,...,\alpha_{N+2}$ for which the module generated by the Gasper-Rahman polynomials becomes reducible. In this case, the semi-infinite dimensional matrices associated with  $\cW^{(N)}_0,\cW^{(N)}_1$ truncate to finite dimensional ones.\vspace{1mm}

\subsection{Raising and lowering operators}
In the context of orthogonal polynomials with one or several variables, the determination of the raising and lowering operators is an important issue. For the  multivariable Macdonald polynomials associated with a root system of type
 $A_{n-1}$, see for instance \cite{KN}. For the Askey-Wilson polynomials, examples of raising and lowering operators 
have been constructed either starting from the $q-$difference-difference bispectral equations or using the one-variable version of the double affine Hecke algebra \cite{sahi}. In this subsection, by analogy with \cite{sahi},  we extend the construction to the multivariable Gasper-Rahman polynomials
(\ref{normAWgen}) in a straighforward manner. Let $p\in{I\!\!N}$ be fixed. Let $V_{p}$  denote the subspace
 of ${\mathbb C}[x_1,x_2,...,x_N]$ generated by the normalized multivariable 
polynomials $\{\hQ^{(N)}(\{n\},\{x\},\{\alpha\})| \ n_1+n_2+...+n_N=p\}$, as in Proposition \ref{blocktri}. 
Below, we introduce $N$ distinct pairs of difference operators $\mathcal{B}^{\pm(k)}_{\{n\}}$ which act as:
\beqa
\mathcal{B}^{\pm (k)}_{\{n\}}\ V_{p}  \subseteq V_{p\pm 1} \quad \mbox{for all} \quad k=1,2,...,N.
\eeqa
Note that raising and lowering operators for the Askey-Wilson polynomials are recovered by setting $N=k=1$ in the formulae below.  
\begin{prop} Let  $\mathcal{W}_0^{(k)},\mathcal{W}_1^{(k)}$ be defined as
(\ref{actWgen}). Define for a given set of integers $\{n\}=(n_1,n_2,...,n_k)$:
\begin{eqnarray}
\qquad \mathcal{B}^{\pm(k)}_{\{n\}} = \left(\frac{\prod_{l=1}^{k-1} 
\big(\mathcal{W}_1^{(l)} - \theta_{\{n+e_l\}}^{*(l)}\big)\big(\mathcal{W}_1^{(l)} - \theta_{\{n-e_l\}}^{*(l)}\big)}{\prod_{l=1}^{k-1} 
\big( \theta_{\{n\}}^{*(l)} -
\theta_{\{n+e_l\}}^{*(l)}\big)\big(\theta_{\{n\}}^{*(l)} -
\theta_{\{n-e_l\}}^{*(l)}\big)}\right)\left( \frac{\big(\mathcal{W}_1^{(k)} -
  \theta_{\{n\mp e_k\}}^{*(k)}\big)\big(\mathcal{W}_1^{(k)} -
    \theta_{\{n\}}^{*(k)}\big)}{\big(\theta^{*(k)}_{\{n\pm e_k\}}-
    \theta^{*(k)}_{\{n\mp e_k\}}\big)\big(\theta^{*(k)}_{\{n\pm e_k\}}-
  \theta^{*(k)}_{\{n\}}\big)}\right) \mathcal{W}_0^{(k)}.\label{B+-} 
\end{eqnarray}
For any $N\geq k$ and any $\{n'\}$ such that
$n'_i=n_i$ for $1\leq i\leq k$, one has
\begin{eqnarray}
\mathcal{B}^{+(k)}_{\{n\}} {\widehat
  Q}^{(N)}(\{n'\},\{x\},\{\alpha\})&=&
b_{n_1\cdots n_k}^{[\overbrace{0\cdots 0 1}^{k\ terms}]} 
E_{n_k}^{+1} {\widehat Q}^{(N)}(\{n'\},\{x\},\{\alpha\})\label{B+}, \\
\mathcal{B}^{-(k)}_{\{n\}} {\widehat
  Q}^{(N)}(\{n'\},\{x\},\{\alpha\})&=&
c_{n_1\cdots n_k}^{[\overbrace{0\cdots 0 -1}^{k\ terms}]} 
E_{n_k}^{-1} {\widehat Q}^{(N)}(\{n'\},\{x\},\{\alpha\}).\label{B-}
\end{eqnarray}
\end{prop}
\begin{proof} Let us first prove Equation (\ref{B+}). In the
  definition (\ref{B+-}) of $\mathcal{B}^{+(k)}_{\{n\}}$,  define $\Gamma$ as the result of the action of
  the $\mathcal{W}_0^{(k)}$ term on ${\widehat Q}^{(N)}(\{n'\},\{x\},\{\alpha\})$, which is determined by
  (\ref{real2}) with (\ref{qdiff3}) for $N\rightarrow k$. 
By construction, through the action of $\mathcal{W}_1^{(1)} -
\theta_{\{n+e_1\}}^{*(1)}$ (resp. $\mathcal{W}_1^{(1)} -
\theta_{\{n-e_1\}}^{*(1)}$) on $\Gamma$, all the terms for
which $\nu_k=1$ (resp. $\nu_k=-1$) vanish. If $\nu_k=0$, a factor of the form
$\theta_{\{n\}}^{*(1)}-\theta_{\{n+e_1\}}^{*(1)}$ (resp. $\theta_{\{n\}}^{*(1)}-\theta_{\{n-e_1\}}^{*(1)}$) appears
in front of the corresponding terms. Therefore, the $l=1$ term in the numerator of
(\ref{B+-}) projects out all the terms in $\Gamma$ for which
$\nu_k=\pm 1$ and leaves the terms for which $\nu_k=0$ unchanged up to a factor
\beqa
\big( \theta_{\{n\}}^{*(1)} -
\theta_{\{n+e_1\}}^{*(1)}\big)\big(\theta_{\{n\}}^{*(1)} -
\theta_{\{n-e_1\}}^{*(1)}\big)\nonumber
\eeqa
in front of each of the remaining terms.
Then, looking at the action of the operators with $l=2$ in the numerator of (\ref{B+-}), which have the same
structure as those for $l=1$, on the remaining $\nu_k=0$ terms of
$\Gamma$, we similarly find that it projects out the terms for which
$\nu_{k-1}-\nu_k=\nu_{k-1}=\pm 1$ and leaves the terms for which $\nu_{k-1}=0$
unchanged up to a factor. 
An immediate induction on $1\leq l\leq k-1$ then shows similarly that
only the terms for which $\nu_i=0$ for $2\leq i \leq k$ are not
projected out, and those are left unchanged up to a factor. Therefore there now only
remains three terms from $\Gamma$: $a^{{[}0\cdots0{]}}$,
$b^{{[}0\cdots0 1{]}}E_{n_k}$ and $c^{{[}0\cdots0 -1{]}}E^{-1}_{n_k}$ acting on ${\widehat Q}^{(N)}(\{n'\},\{x\},\{\alpha\})$.
Finally, the second part of  (\ref{B+-}) acts on the remaining expression as follows: $(\mathcal{W}_1^{(k)} -
\theta_{\{n\}}^{*(k)})/(\theta^{*(k)}_{\{n+ e_k\}}- \theta^{*(k)}_{\{n\}})$
projects out the term in $a^{{[}0\cdots0{]}}$, whereas $(\mathcal{W}_1^{(k)} -
\theta_{\{n-e_k\}}^{*(k)})/(\theta^{*(k)}_{\{n+ e_k\}}- \theta^{*(k)}_{\{n-
  e_k\}})$ projects out the term in $c^{{[}0\cdots0-1{]}}E^{-1}_{n_k}$, and
again both leave the term in $b^{{[}0\cdots01{]}}E_{n_k}$ unchanged.
Thus, dividing the final expression by
\beqa
\prod_{l=1}^{k-1} 
\big( \theta_{\{n\}}^{*(l)} -
\theta_{\{n+e_l\}}^{*(l)}\big)\big(\theta_{n}^{*(l)} -
\theta_{\{n-e_l\}}^{*(l)}\big),\nonumber
\eeqa
 in the end, only one term has survived from $\Gamma$. It is given in
 the r.h.s of (\ref{B+}). The proof of the other equation (\ref{B-}) is similar.
%
\end{proof}

Let us mention that in the context of quantum integrable systems, the solution of the `inverse problem' - i.e. the explicit relation between local operators and the elements of the non-Abelian algebra that ensures the integrability of the model - plays a central role in the analysis of correlation functions, for instance. The explicit construction of the raising and lowering 
operators  (\ref{B+-}) in terms of the fundamental elements of the $q-$Onsager algebra provides such a solution. It should find applications in the analysis of models generated from the $q-$Onsager algebra, for instance the open XXZ spin chain. 
\vspace{2mm}

\section{Bispectrality and the relation with tridiagonal pairs}
As mentionned in the introduction, for the $q-$Askey scheme of
orthogonal polynomials in one variable the bispectral property has a
natural interpretation within the framework of Leonard's theorem \cite{Leon} and Leonard pairs \cite{Terw0,Terw}. Namely, for a pair of elements $A,A^*$ that act on an irreducible finite dimensional module, are diagonalizable, and satisfy the so-called Askey-Wilson relations, it is known that there exists two basis with respect to which the two matrices representing $A,A^*$ are diagonal (resp. tridiagonal) and tridiagonal (resp. diagonal). In such examples, the overlap coefficients between the two basis are given by the Askey-Wilson polynomials evaluated on a discrete support. First examples appeared, for instance, in \cite{Zhed}.\vspace{1mm}

Tridiagonal pairs are generalizations of Leonard pairs \cite{Ter0}, in which case one assumes that degeneracies may occur in the spectrum of the two matrices under consideration. Provided the elements $A,A^*$ satisfy a pair of relations called the tridiagonal relations (see (\ref{TD1}), (\ref{TD2}) below), it implies the following: in the basis in which the first matrix is diagonal with degeneracies, the other matrix takes a block tridiagonal structure. Furthermore, there exists another basis with respect to which the first matrix transforms into a block tridiagonal matrix, whereas the second one transforms into a diagonal matrix with degerenacies. In the literature, it is now well-understood that the representation theory of tridiagonal algebras - in particular the $q-$Onsager algebra (\ref{qDG}) - is intimately connected with the theory of tridiagonal pairs \cite{Ter03}, as we will recall below.
For instance, finite dimensional irreducible modules of the
$q-$Onsager algebra for which $A$ and $A^*$ are
  diagonalizable have 
been classified in \cite{ITaug} for $q$ not a root of unity. As an application,
 tridiagonal pairs of $q-$Racah type (see  `case I' below) over ${\mathbb C}$ have been 
classified for $q$ not a root of unity. For an algebraically closed field and no
 restrictions on $q$, note that a classification of tridiagonal pairs is given in \cite{INT} (see also \cite{IT2}). 
 However, to our knowledge, the connection with the theory of orthogonal polynomials has remained, in general, an open problem. The purpose of this Section is to show that the bispectral property of the Gasper-Rahman polynomials finds a natural interpretation in the theory of tridiagonal pairs. From that point of view, given a TD pair, the overlap coefficients between the two basis are given by  the normalized Gasper-Rahman polynomials (\ref{normAWgen}) evaluated on a discrete support \cite{F11}. In this context, a correspondence between the finite dimensional tensor product evaluation modules of the $q-$Onsager algebra considered in \cite{Bas3,ITaug} and the factorized structure of Gasper-Rahman polynomials (\ref{AWgen}) arises naturally.

\subsection{Tridiagonal pairs and tridiagonal algebras}
In the present subsection, we recall  the results taken from \cite{Ter93,Ter01,Ter03,NT:muqrac}.
Let ${\mathbb K}$ be an arbitrary field unless otherwise noted. Let $U$ denote a vector space over $\K$ with finite
positive dimension. For a  linear transformation $A:U\to U$
and a subspace $W \subseteq U$, we call $W$ an
eigenspace of $A$ whenever $W\not=0$ and there exists $\theta \in \K$ such that 
$W=\lbrace u \in U \ | \;Au = \theta u\rbrace$. $\theta$ is the eigenvalue of
$A$ associated with $W$. $A$ is diagonalizable whenever $U$ is spanned by the eigenspaces of $A$.

\begin{defn}[See \cite{Ter01}, Definition 2.1]
\label{tdp}
\rm
Let $U$ denote a vector space over $\K$ with finite
positive dimension. 
By a  tridiagonal pair (or $TD$ pair)
on $U$
we mean an ordered pair of linear transformations
$A:U \to U$ and 
$A^*:U \to U$ 
that satisfy the following four conditions.
\begin{itemize}
\item[(i)] Each of $A,A^*$ is diagonalizable on $U$.
\item[(ii)] There exists an ordering $\lbrace U_p\rbrace_{p=0}^d$ of the  
eigenspaces of $A$ such that 
\begin{equation}
A^* U_p \subseteq U_{p-1} + U_p+ U_{p+1} \qquad \qquad 0 \leq p \leq d,
\label{eq:t1}
\end{equation}
where $U_{-1} = 0$ and $U_{d+1}= 0$.
\item[(iii)] There exists an ordering $\lbrace U^*_p\rbrace_{p=0}^{\delta}$ of
the  
eigenspaces of $A^*$ such that 
\begin{equation}
A U^*_p \subseteq U^*_{p-1} + U^*_p+ U^*_{p+1} 
\qquad \qquad 0 \leq p \leq \delta,
\label{eq:t2}
\end{equation}
where $U^*_{-1} = 0$ and $U^*_{\delta+1}= 0$.
\item[(iv)] There does not exist a subspace $W$ of $U$ such  that $AW\subseteq W$,
$A^*W\subseteq W$, $W\not=0$, $W\not=U$.
\end{itemize}
We say the pair $A,A^*$ is  over $\K$.
We call $U$ the 
underlying
 vector space.
\end{defn}

\begin{note} According to a common 
notational convention, for a linear transformation $A$
the conjugate-transpose of $A$ is denoted
$A^*$. We are not using this convention here.
In a TD pair $A,A^*$, the linear transformations 
$A$ and $A^*$ are arbitrary subject to (i)-(iv) above.
\end{note}

Let $A,A^*$ denote a TD pair on $U$, as in Definition 
\ref{tdp}. By \cite[Lemma~4.5]{Ter01} the integers $d$ and $\delta$ from
(ii), (iii) are equal and called the diameter of the
pair. An ordering of the eigenspaces of $A$ (resp. $A^*$)
is said to be {\em standard} whenever it satisfies  (\ref{eq:t1})
 (resp. (\ref{eq:t2})).  Let $\{U_p\}_{p=0}^d$ (resp.
$\{U^*_p\}_{p=0}^d$) denote a standard ordering of the eigenspaces
 of $A$ (resp. $A^*$). For $0 \leq p \leq d$, let 
$\theta_p$  (resp. $\theta^*_p$) denote the eigenvalue of
$A$  (resp.  $A^*$) associated with $U_p$  (resp. $U^*_p$). According
 to \cite[Theorem~4.4]{Ter03}, for ${\mathbb K}={\mathbb C}$ the
 eigenvalues $\theta_p$ and $\theta^*_p$, $p=0,1,...,d$, can take
 three different forms:
\vspace{1mm}

$\bullet$ Case I: the $q-$Racah type (with $b,c,b^*,c^*\not=0$) 
\begin{eqnarray}
\label{eq:const1}
&&\theta_p = a + b q^{2p-d} + c q^{d-2p}, \nonumber
\\
\label{eq:const2}
&&\theta^*_p = a^* + b^* q^{2p-d} + c^* q^{d-2p};\nonumber
\end{eqnarray}

$\bullet$ Case II: the Racah type (with $c,c^*\not=0$) 
\beqa
\theta_p &=& a + bp + cp(p-1)/2,  \nonumber\\
\theta^*_p &=& a^* + b^*p + c^*p(p-1)/2;  \nonumber
\eeqa

$\bullet$ Case III: the Bannai-Ito case
\begin{eqnarray}
\label{eq:const3}
&&\theta_p = a + b(-1)^p  + cp (-1)^p,\nonumber
\\
&&\theta^*_p = a^* +  b^*(-1)^p  + c^* p (-1)^p\nonumber.
\end{eqnarray}

Here  $a,a^*, b,\;b^*,\ c ,\;c^*$ are scalars and $q\not=0, \ q^2 \not=1,
\ q^2 \not=-1$. In this case,  we say that $A,A^*$ is  a tridiagonal pair of $q-$Racah (case I) or Racah (case II) type \cite[Theorem~5.3]{NT:muqrac}. Note that the cases listed above are the `most general' cases in which the main free parameters are nonzero. Other examples for which some of the parameters are set to zero can be found in \cite{Ter033}.\vspace{2mm}

The theory of tridiagonal pairs is closely related with the representation theory of tridiagonal algebras. Tridiagonal algebras are defined as follows:

\begin{defn}[See \cite{Ter03}, Definition 3.9]\label{TDgendef}
Let  $\beta, \gamma, \gamma^*, \rho, \rho^* $ denote scalars in $\K$. We define 
 $T=T(\beta, \gamma, \gamma^*,
\rho, \rho^*)$  as the associative $\K$-algebra with unit  generated by two elements ${\textsf A}$, ${\textsf A}^*$
subject to the relations 
\beqa
\lbrack {\textsf A},{\textsf A}^2{\textsf A}^*-\beta {\textsf A}{\textsf A}^*{\textsf A} + {\textsf A}^*{\textsf A}^2 -\gamma ({\textsf A}{\textsf A}^*+{\textsf A}^*{\textsf A})-\rho {\textsf A}^*\rbrack 
&=&0,
\label{TD1}
\\
\lbrack {\textsf A}^*,{\textsf A}^{*2}{\textsf A}-\beta {\textsf A}^*{\textsf A}{\textsf A}^* + {\textsf A}{\textsf A}^{*2} -
\gamma^* ({\textsf A}^*{\textsf A}+{\textsf A}{\textsf A}^*)-\rho^* {\textsf A} \rbrack
&=&0.
\label{TD2}
\eeqa
We refer to $T$ as the tridiagonal algebra 
(or TD algebra) over $\K$ with parameters $\beta, \gamma, \gamma^*, \rho, \rho^* $.
 ${\textsf A}$ and ${\textsf A}^*$ are called the standard generators of $T$.
\end{defn}
In the literature, note that the relations (\ref{TD1}),(\ref{TD2}) are
called the tridiagonal (TD) relations. In case I, if
$\beta = q^2+q^{-2}$, $\gamma=\gamma^*=0$, $\rho=\rho^*\neq 0$, the parameter sequence is said to be  `reduced'. In this case, the TD relations coincide with the defining relations of the $q-$Onsager algebra (\ref{qDG}) with, for instance, ${\textsf W}_0\rightarrow {\textsf A}$, ${\textsf W}_1\rightarrow {\textsf A}^*$.

\begin{rem} For $\beta  = 2$, $\gamma = \gamma^*=0$, $\rho =\rho^* = 16$ the TD relations coincide with the Dolan-Grady relations \cite{DG}. For $\beta = q^2+q^{-2}$, $\gamma=\gamma^*=0$, $\rho=\rho^*=0$, the TD relations coincide with the 
$q-$Serre relations of $U_q(\widehat{sl_2})$.  
\end{rem}

It is known that every TD pair comes from an irreducible  finite dimensional module of a tridiagonal algebra. The following theorems establish an explicit connection between the representation theory of tridiagonal algebras (finite dimensional irreducible modules for $q$ not a root of unity) and the theory of tridiagonal pairs, as defined in Definition \ref{tdp}.
\begin{thm}[See \cite{Ter03}, Theorem 3.7]\label{thmTDqOA}
Let  $A, A^*$ denote a TD pair
over $\K$.
Then
there exists a sequence of  
 scalars
$\beta, \gamma, \gamma^*, \rho, \rho^* $ taken from $\K$
such that $A,A^*$ satisfy the tridiagonal relations (\ref{TD1}), (\ref{TD2}). 
The parameter sequence $\beta, \gamma, \gamma^*, \rho, \rho^* $ is uniquely determined by the pair if the diameter is at least 3.
\end{thm}

\begin{example}[See \cite{Ter03}, Lemma 4.5] Let $A,A^*$ be a TD pair of $q-$Racah type (case I). Suppose $a=a^*=0$ and $bc=b^*c^*$.   Then, $A,A^*$ satisfy the defining relations of the $q-$Onsager algebra (\ref{qDG}) with ${\textsf W}_0\rightarrow {A}$, ${\textsf W}_1\rightarrow {A}^*$, $\rho=-bc(q^2-q^{-2})^2$.
\end{example}

\begin{thm}[See \cite{Ter03}, Theorem 3.10]\label{thm32}
Let  $\beta, \gamma, \gamma^*,\rho, \rho^*$ denote scalars in $\K$, and 
assume $q$ is not a root of unity, where $q^2+q^{-2} = \beta$. 
Let $T$ denote the TD algebra over $\K$  with parameters $\beta, \gamma, \gamma^*,
\rho, \rho^*$ and standard generators $A, A^*$. Let $U$ denote an irreducible finite dimensional $T$-module
and assume each of $A, A^*$ is diagonalizable on $U$. Then $A, A^*$ act on $U$ as a TD pair. 
\end{thm}

According to some observations in \cite[p. 11]{Ter03}, let $\tilde{A},\tilde{A}^*$ denote a TD pair on $U$ with parameter sequence $\beta, \gamma, \gamma^*, \tilde{\rho},  \tilde{\rho}^*$. Let $I$ denote the identity and assume $\beta\neq 2$. Then, the TD pair
\beqa
\qquad A=\rho^{\frac{1}{2}}\left(\tilde{\rho} + \frac{\gamma^2}{2-\beta}\right)^{-\frac{1}{2}}\left(\tilde{A} + \gamma(\beta-2)^{-1}I\right) \ ,\qquad A^*=\rho^{\frac{1}{2}}\left(\tilde{\rho}^* + \frac{{\gamma^*}^2}{2-\beta}\right)^{-\frac{1}{2}} \left(\tilde{A}^* + \gamma^*(\beta-2)^{-1}I \right)\label{reduc}
\eeqa
has a {\it reduced} parameter sequence: the TD pair (\ref{reduc}) satisfy the defining relations of the $q-$Onsager algebra (the so-called $q-$Dolan-Grady relations) (\ref{qDG}) with ${\textsf W}_0\rightarrow {A}$, ${\textsf W}_1\rightarrow {A}^*$. It is easy to show that such transformation is invertible. As a consequence, provided $\beta\neq 2$, finite dimensional irreducible modules of the tridiagonal algebra with $\gamma,\tilde{\rho}$ (resp.  $\gamma,\tilde{\rho}$) both not $0$
and of the $q-$Onsager algebra are isomorphic. For this reason, in the
next subsection we simply focus on the construction of finite
dimensional  modules. The special case $q=1$  for which $\beta=2$ will be treated separately in Section 5.


\subsection{Finite dimensional modules}
In view of the homomorphism (\ref{actWgen}) for $k=N$ together with (\ref{real2}), a comparison between the structure of the spectrum of a tridiagonal pair of $q-$Racah type for $q\neq 1$ (called `Case I' above) and the bispectral problem (\ref{qdiffgen}), (\ref{recgen}) strongly suggests that  finite dimensional modules of the $q-$Onsager algebra (\ref{qDG}) can be derived by studying certain quotients of the infinite dimensional modules  based on the normalized Gasper-Rahman multivariable polynomials  (\ref{normAWgen}), considered in the previous Section. \vspace{1mm}

Considering bispectral relations of the form (\ref{biAW}), (\ref{biAW2}), it is well-known that finite dimensional modules of dimension $2j+1$ based on the Askey-Wilson polynomials (\ref{awpoly}) can be easily constructed provided the parameters  $a,b,c,d$  are chosen such that the coefficient $b_{2j}=0$ \cite{F12},  and that $z$ is restricted to a discrete support. In this case, the $q-$difference equation (\ref{biAW}) becomes a three-term recurrence relation in addition to  (\ref{biAW2}), and both recurrence relations terminate. Then, the corresponding bispectral problem is interpreted as the spectral problem associated with a Leonard pair. Such a technique can easily be extended in order to construct finite dimensional modules based on the normalized Gasper-Rahman polynomials (\ref{normAWgen}), in view of the factorized structure  (\ref{AWgen}). To this end, restrictions on the parameters $\alpha_1,...,\alpha_{N+1}$ and a discrete support of the variables $z_1,z_2,...,z_N$ are considered. \vspace{1mm}

Let $\textsf{W}_0,\textsf{W}_1$ be the standard generators of the $q-$Onsager algebra (\ref{qDG}) with $q\neq 1$. Recall that
$\textsf{W}_0,\textsf{W}_1$ as defined in Proposition \ref{realqDG2} induce a module structure on ${\mathbb C}[ x_1,x_2,...,x_N]$. To construct finite dimensional  modules of the $q-$Onsager algebra, we proceed as follows. Let $\tilde{n}_1,\tilde{n}_2,...,\tilde{n}_N$ be positive integers. Assume the $N-$variables are now restricted to the following values:
\beqa
 z_k=\frac{\alpha_{N+1}\alpha_{N+2}}{\alpha_k}q^{2\tilde{\kN}_{N+1-k}}\quad \mbox{with} \quad \tilde{\kN}_k=\tilde{n}_1+\tilde{n}_2+\cdots+\tilde{n}_k\quad \mbox{for}\quad k=1,2,...,N.\label{zk}
\eeqa
Then, for the following discussion we introduce the notation:
\beqa
\tfQ^{(N)}(\{n\},\{\tilde{n}\},\{\alpha\})=  \hQ^{(N)}(\{n\},\{x\},\{\alpha\})|_{z_k=\frac{\alpha_{N+1}\alpha_{N+2}}{\alpha_k}q^{2\tilde{\kN}_{N+1-k}},k=1,2,...,N}.\label{polyrec}
\eeqa

According to the results of the previous Section, on the restricted
support generated by (\ref{polyrec}), one induces
operators having the same action on
$\tfQ^{(N)}(\{n\},\{\tilde{n}\},\{\alpha\})$ as $x_1$ and the same
action on $\hQ^{(N)}(\{n\},\{x\},\{\alpha\})$ as   $\frac{1}{2}{\mathbb D}_{\{z\}}^{*(N)}$. For simplicity, we will not
  change their names and identify them again with 
 $\textsf{W}_0,\textsf{W}_1$ (see
  Proposition \ref{condpargen} below). They
are still diagonalizable. The (degenerate) eigenvalues of $\textsf{W}_0$ (resp.$\textsf{W}_1$) are, respectively, given by:
\beqa
\qquad \theta^{(N)}_{\{\tilde{n}\}}= \frac{1}{2}\left(\frac{\alpha_1}{\alpha_{N+1}\alpha_{N+2}}q^{-2\tilde{\kN}_N} +   \frac{\alpha_{N+1}\alpha_{N+2}}{\alpha_1} q^{2\tilde{\kN}_N}\right),\quad  
\theta^{*(N)}_{\{n\}}= \frac{1}{2}\left(   \frac{\alpha_0q}{\alpha_{N+1}} q^{-2\kN_N} + \frac{\alpha_{N+1}}{\alpha_0q} q^{2\kN_N}\right)   \label{spect},
\eeqa
and the corresponding eigenfunctions are $\tfQ^{(N)}(\{n\},\{\tilde{n}\},\{\alpha\})$.

Then, the vector space is spanned by the eigenspaces of either $\textsf{W}_0$ or $\textsf{W}_1$. Let $p\in{\mathbb N}$ be fixed. With respect to $\textsf{W}_0,\textsf{W}_1$, two different families of ordered eigenspaces can be constructed, respectively, from
\beqa
{\cal V}_p^{(N)}&=&\mathrm{Span}\{\tfQ^{(N)}(\{n\},\{\tilde{n}\},\{\alpha\})|\tilde{n}_1+...+\tilde{n}_N=p\},  \quad \{n\} \ \mbox{fixed},\nonumber\\
\mbox{or} \qquad {\cal V}_p^{*(N)}&=&\mathrm{Span}\{\tfQ^{(N)}(\{n\},\{\tilde{n}\},\{\alpha\})|n_1+...+n_N=p\}, \quad \{\tilde{n}\}\  \mbox{fixed}.\nonumber
\eeqa

Up to now, note that the set of integer variables $\{n\},\{\tilde{n}\}$ has not been restricted. Now,
provided the parameters $\alpha_1,...,\alpha_{N+1}$ satisfy certain conditions,  finite dimensional modules can be constructed as follows.
\begin{prop}\label{condpargen} Let $j_1,j_2,...,j_N$ be positive half-integers or integers and define $d_N=2(j_1+j_2+...+j_N)$.  
Assume there exists parameters $\alpha_1,\alpha_2,...,\alpha_{N+1}$ such that:
\beqa
\frac{\alpha_{k+1}}{\alpha_{k}}=q^{-2j_{k}}\quad k=1,2,...,N.\label{condp}
\eeqa
Then, the maps defined by 
\beqa
\textsf{W}_0  \mapsto   x_1, \qquad \textsf{W}_1  \mapsto \frac{1}{2}{\mathbb D}_{\{z\}}^{*(N)}, \qquad \rho \mapsto -\frac{(q^2-q^{-2})^2}{4} ,
\eeqa
yield a family of homomorphisms from $O_q(\widehat{sl_2})$ to any of the vector
spaces ${\cal V}^{*\{\tilde{n}\}}=  \bigoplus_{p=0}^{d_N}{\cal V}_p^{*(N)}$
  with $\tilde{n}$ fixed, and also to any ${\cal V}^{\{n\}}=
  \bigoplus_{p=0}^{d_N}{\cal V}_p^{(N)}$ with $n$ fixed.
These are finite dimensional  modules of the $q-$Onsager algebra (\ref{qDG}).
\end{prop}
Note that, as stated at the top of this page, although $x_1$ and
  $\frac{1}{2}{\mathbb D}_{\{z\}}^{*(N)}$ are actually distinct
  operators for each of the target vector spaces, we leave the same
  notation for all of them since they are induced from the same operator.

\begin{proof} 
 First, we have to show that ${\cal V}^{*\{\tilde{n}\}}$ is a finite dimensional  module of the $q-$Onsager algebra (\ref{qDG}). According to the results of the previous Section, on the space generated by (\ref{polyrec}) the action of the  standard generators $\textsf{W}_0,\textsf{W}_1$ is characterized by  (\ref{real2}). Observe that the element $\textsf{W}_1$ acts on $\bigoplus_{p=0}^{d_N}{\cal V}_p^{*(N)}$  as a diagonal matrix with entries $\theta^{*(N)}_{\{n\}}$ such that $\kN_N=p$. Given the ordering of the eigenspaces $\{{\cal V}_p^{*(N)}\}_{p=0}^{d_N}$ of $\textsf{W}_1$, next we consider the action of the element $\textsf{W}_0$, using (\ref{real2}) with (\ref{qdiff3}). Observe that  $\textsf{W}_0$ acts as (\ref{eq:t2}). 
To prove that the module ${\cal V}^{*\{\tilde{n}\}}$ is finite dimensional, it is sufficient to show that the coefficients (\ref{bcoeff}),(\ref{ccoeff}),(\ref{acoeff}) are such that:
\beqa
\qquad b_{n_1 n_2 \cdots n_N}^{[\nu_N\ \nu_{N-1}-\nu_N\ \cdots \nu_2-\nu_3\ 1-\nu_2]}    &=& 0 \quad \mbox{if} \quad 
\nu_N=1\quad \mbox{and} \quad n_1=2j_1, 
 \label{conditions}\\
&&\qquad \quad \nu_{N-k}- \nu_{N+1-k}=1\quad \mbox{and} \quad n_{k+1}=2j_{k+1},\nonumber\\
&& \qquad \quad \nu_2=0 \quad \mbox{and} \quad  n_{N}=2j_N,\ \  
 \nu_2=-1 \quad \mbox{and} \quad  n_{N}=2j_N,2j_N-1, \nonumber\\
&&\qquad \quad \nu_{N-k}- \nu_{N+1-k}=2\quad \mbox{and} \quad n_{k+1}=2j_{k+1}, 2j_{k+1}-1,\nonumber\\
c_{n_1 n_2 \cdots n_N}^{[\nu_N\ \nu_{N-1}-\nu_N\ \cdots \nu_2-\nu_3\ -1-\nu_2]}    &=& 0  \quad \mbox{if} \quad 
\nu_N=1\quad \mbox{and} \quad n_1=2j_1, 
\nonumber\\
&&\qquad \quad \nu_{N-k}- \nu_{N+1-k}=1\quad \mbox{and} \quad n_{k+1}=2j_{k+1},\nonumber\\
&&\qquad \quad \nu_{N-k}- \nu_{N+1-k}=2\quad \mbox{and} \quad n_{k+1}=2j_{k+1},2j_{k+1}-1,\nonumber\\
a_{n_1 n_2 \cdots n_N}^{[\nu_N\ \nu_{N-1}-\nu_N\ \cdots \nu_2-\nu_3\ -\nu_2]}   &=& 0 \quad \mbox{if} \quad 
\nu_N=1\quad \mbox{and} \quad n_1=2j_1,
\nonumber\\
&&\qquad \quad \nu_{N-k}- \nu_{N+1-k}=1\quad \mbox{and} \quad n_{k+1}=2j_{k+1},\nonumber\\
&&\qquad \quad \nu_2=-1 \quad \mbox{and} \quad  n_{N}=2j_N,\nonumber\\
&&\qquad \quad \nu_{N-k}- \nu_{N+1-k}=2\quad \mbox{and} \quad n_{k+1}=2j_{k+1},2j_{k+1}-1,\nonumber
\eeqa
for all $k=1,2,...,N-2$. 
%
%
To show that, one uses the formulae given in the Appendix: given the conditions on the ratios of the parameters (\ref{condp}), one finds:
\beqa
\kb(B_k^{1,0}(z))&=&0 \quad \mbox{and}\quad \kb(B_k^{0,-1}(z))=0  \quad \mbox{for}\quad n_{N+1-k}=2j_{N+1-k},\quad k=1,...,N-1,\nonumber\\
\kb(B_N^{1,0}(z))&=&0  \quad \mbox{for}\quad n_1=2j_1,\nonumber\\
\kb(B_k^{1,-1}(z))&=&0  \quad \mbox{for}\quad n_{N+1-k}=2j_{N+1-k},\ 2j_{N+1-k}-1, \qquad k=1,...,N-1.\nonumber
\eeqa
As a consequence, using (\ref{Cnu}) the coefficients (\ref{bcoeff}),(\ref{ccoeff}),(\ref{acoeff}) automatically satisfy  (\ref{conditions}).\vspace{1mm}

Secondly, we have to show that ${\cal V}^{\{n\}}$ is also a finite dimensional  module of the $q-$Onsager algebra (\ref{qDG}).
 Recall (\ref{zk}). On one hand, according to  (\ref{actWgen}) for $k=N$,  the element ${\textsf W}_0$ acts 
on $\bigoplus_{p=0}^{d_N}{\cal V}_p^{(N)}$ as a diagonal matrix with entries $\theta^{(N)}_{\{\tilde{n}\}}$. On the other hand, introduce the `dual' difference operator:
\beqa
{\mathbb D}_{\{\tilde{n}\}}^{*(N)}&=&{\mathbb D}_{\{z\}}^{*(N)}|_{\{z_k=\frac{\alpha_{N+1}\alpha_{N+2}}{\alpha_k}q^{2\tilde{\kN}_{N+1-k}},\ k=1,2,...,N\}}\nonumber\\
&=&{\mathbb D}_{\{n\}}^{(N)}|_{\{\alpha_0\rightarrow \alpha_0;\ \ \alpha_k \rightarrow \frac{\alpha_0\alpha_{N+1}\alpha_{N+2}}{\alpha_{N+2-k}}, \ k=1,...,N+1;\ \ \alpha_{N+2}\rightarrow \frac{\alpha_1}{\alpha_0q};\ \ n\rightarrow \tilde{n}\}}.\label{dualD}
\eeqa
Then, with respect to the indices $\{\tilde{n}\}$ it follows that  $\textsf{W}_1$  act on ${\cal V}^{\{n\}}$ according to (\ref{eq:t1}). Now, using the explicit expressions of the coefficients (\ref{Cnu}), one finds that the coefficients of  `$b-$type' characterizing (\ref{dualD}) are vanishing for $\tilde{n}_k= j_{N+1-k}$ with $k=1,2,...,N$.  As a consequence, $\textsf{W}_0,\textsf{W}_1$ act on  ${\cal V}^{\{n\}}$ as a diagonal matrix and an  block tridiagonal matrix, respectively.  The claim follows.
\end{proof}


\begin{rem} Note that 
\beqa
dim\left(\bigoplus_{p=0}^{d_N}{\cal V}_p^{(N)}\right)= dim\left(\bigoplus_{p=0}^{d_N}{\cal V}_p^{*(N)}\right)=\prod_{k=1}^N(2j_k+1).\label{dimevalmod}
\eeqa
\end{rem}

\vspace{2mm}

\begin{example} For $N=1$, let $j$ be a positive half-integer or integer. Define 
\beqa
{\cal V}_{n_1}^{*(1)}&=&\mathrm{Span}\{\tfQ^{(1)}(n_1,\tilde{n}_1,\alpha_0,\alpha_1,\alpha_2,\alpha_3)\},\quad 
\ \tilde{n}_1 \ \mbox{fixed},\nonumber\\
{\cal V}_{\tilde{n}_1}^{(1)}&=&\mathrm{Span}\{\tfQ^{(1)}(n_1,\tilde{n}_1,\alpha_0,\alpha_1,\alpha_2,\alpha_3)\}, \quad 
\ n_1 \ \mbox{fixed},\nonumber
\eeqa
with 
\beqa
\alpha_0=q^{\beta^*-\beta-1}, \quad \alpha_1=q^{-\beta}, \quad \alpha_2=q^{-2j-\beta}, \quad \alpha_3=q^{-\beta},\label{p1} 
\eeqa
where $\beta,\beta^*$ are arbitrary scalars. The finite dimensional module ${\cal V}^{n_1}$ (resp. ${\cal V}^{*\tilde{n}_1}$)
of the $q-$Onsager algebra has dimension $2j+1$. The eigenvalues of $\textsf{W}_0,\textsf{W}_1$, are distinct and given by:
\beqa
\theta^{(1)}_{\tilde{n}_1}= \frac{1}{2}\left(q^{\beta+2j -2\tilde{n}_1} + q^{-\beta-2j +2\tilde{n}_1}\right),\quad  
\theta^{*(1)}_{n_1}= \frac{1}{2}\left(q^{\beta^*+2j -2n_1} + q^{-\beta^*-2j +2n_1}\right)   \nonumber.
\eeqa

\end{example}

Note that using (\ref{p1}) in (\ref{AWdiff}), one can check that $b_{2j_1}=0$ and $c_0=0$.

\vspace{2mm}

\begin{example} For $N=2$, let $j_1,j_2$ be positive half-integers or integers. Define 
\beqa
{\cal V}_{p}^{*(2)}&=&\mathrm{Span}\{\tfQ^{(2)}(n_1,n_2,\tilde{n}_1,\tilde{n}_2,\alpha_0,\alpha_1,\alpha_2,\alpha_3,\alpha_4)| n_1+n_2=p\},  \quad
\tilde{n}_1,\tilde{n}_2 \ \mbox{fixed},\nonumber\\
{\cal V}_{p}^{(2)}&=&\mathrm{Span}\{\tfQ^{(2)}(n_1,n_2,\tilde{n}_1,\tilde{n}_2,\alpha_0,\alpha_1,\alpha_2,\alpha_3,\alpha_4)| \tilde{n}_1+\tilde{n}_2=p\},
\quad n_1,n_2  \ \mbox{fixed},\nonumber
\eeqa 
with 
\beqa
\alpha_0=q^{\beta^*-\beta-1}, \quad \alpha_1=q^{-\beta}, \quad \alpha_2=q^{-2j_1-\beta}, \quad \alpha_3=q^{-2(j_1+j_2)-\beta},\quad \alpha_4=q^{-\beta},\label{p2} 
\eeqa
where $\beta,\beta^*$ are arbitrary scalars. The finite dimensional module ${\cal V}^{n_1,n_2}$ (resp. ${\cal V}^{*\tilde{n}_1,\tilde{n}_2}$) of the $q-$Onsager algebra has dimension $(2j_1+1)(2j_2+1)$. The $2(j_1+j_2)+1$  degenerate eigenvalues  of $\textsf{W}_0,\textsf{W}_1$,  are given by
\beqa
\theta^{(2)}_{\tilde{n}_1\tilde{n}_2}= \frac{1}{2}\left(q^{\beta+2(j_1+j_2) -2\tilde{\kN}_2} + q^{-\beta-2(j_1+j_2) +2\tilde{\kN}_2}\right),\quad  
\theta^{*(2)}_{n_1 n_2}= \frac{1}{2}\left(q^{\beta^*+2(j_1+j_2) -2\kN_2} + q^{-\beta^*-2(j_1+j_2) +2\kN_2}\right)   \nonumber.
\eeqa
\end{example}

Note that using (\ref{p2}) in (\ref{AW2}), for any $n_1,n_2$ one can check that:
\beqa
&&c^{[-1 0]}_{0  n_2}=0,\quad c^{[0 -1]}_{n_1 0}=0,\quad c^{[1-2]}_{n_1 0}=0,\quad c^{[1-2]}_{n_1  1}=0, \quad c^{[1-2]}_{2j_1 n_2}=0,\nonumber\\
&&b^{[1 0]}_{2j_1  n_2}=0,\quad b^{[0 1]}_{n_1 2j_2}=0,\quad b^{[-12]}_{n_1 2j_2}=0,\quad b^{[-12]}_{n_1  2j_2-1}=0, \quad b^{[-12]}_{0 n_2}=0,\nonumber\\
&&a^{[1 -1]}_{n_1  0}=0,\quad a^{[-1 1]}_{0 n_2}=0,\quad a^{[1-1]}_{2j_1 n_2}=0,\quad a^{[-11]}_{n_1  2j_2}=0.\nonumber
\eeqa 

\vspace{1mm}

\begin{example}\label{modXXZ} For $N$ generic, let $j_k=1/2$ for all $k\in\{1,2,...,N\}$. Define 
\beqa
{\cal V}_{p}^{*(N)}&=&\mathrm{Span}\{\tfQ^{(N)}(\{n\},\{\tilde{n}\},\{\alpha\})| \kN_N=p\},  \quad
 \ \tilde{n}_k \ \ \mbox{fixed},\nonumber\\
{\cal V}_{p}^{(N)}&=&\mathrm{Span}\{\tfQ^{(N)}(\{n\},\{\tilde{n}\},\{\alpha\})| \tilde{\kN}_N=p\},  \quad 
\ n_k \ \ \mbox{fixed},\quad k=1,2,...,N,\nonumber
\eeqa 
with
\beqa
\alpha_0=q^{\beta^*-\beta-1}, \quad \alpha_1=q^{-\beta},\quad \alpha_{k+1} = q^{-k-\beta} \quad \mbox{for}\quad k=1,...,N,\quad \alpha_{N+2}=q^{-\beta},\label{pN12}
\eeqa
where $\beta,\beta^*$ are arbitrary scalars. The finite dimensional module ${\cal V}^{\{n\}}$ (resp. ${\cal V}^{*\{\tilde{n}\}}$) of the $q-$Onsager algebra has dimension $2^N$. The $N+1$  degenerate eigenvalues  of $\textsf{W}_0,\textsf{W}_1$,  are given by
\beqa
\theta^{(N)}_{\{\tilde{n}\}}= \frac{1}{2}\left(q^{\beta+N -2\tilde{\kN}_N} + q^{-\beta-N +2\tilde{\kN}_N}\right),\quad  
\theta^{*(N)}_{\{n\}}= \frac{1}{2}\left(q^{\beta^*+N -2\kN_N} + q^{-\beta^*-N +2\kN_N}\right)   \nonumber.
\eeqa
\end{example}

\vspace{1mm}

As mentionned at the beginning of this Section, finite dimensional irreducible modules of the $q-$Onsager algebra (\ref{qDG}) were classified in \cite[Theorem 1.15]{ITaug}, using an embedding of the $q-$Onsager algebra
into the $U_q({\cal L}(sl_2))$-loop algebra \cite{F13} \cite[Proposition 1.13]{ITaug}. In the analysis, tensor product evaluation modules of the $q-$Onsager algebra  \cite[Section 1.4]{ITaug}  (see also \cite[Lemma 2.5]{Bas3})
\beqa
V(2j_N,v_N)\otimes \cdots \otimes V(2j_1,v_1) \label{modeval}
\eeqa
were considered, where $v_j\in \mathbb C$ are usually called the evaluation parameters. From a general point of view, for $q$ not a root of unity a TD pair of $q-$Racah type is afforded via the 
embedding of the $q-$Onsager algebra into the $U_q({\cal L}(sl_2))$-loop algebra by (\ref{modeval}).
 It has diameter $d=2j_1+...+2j_N$ and the dimension of
the vector space on which $\textsf{W}_0,\textsf{W}_1$ act follows from \cite[Proposition 1.24]{ITaug} with $\lambda\rightarrow 1,\ell_i\rightarrow 2j_i,n\rightarrow N$, which result coincides \cite{F14} with (\ref{dimevalmod}). Note the factorized
structure of the Gasper-Rahman polynomials (\ref{AWgen})  by analogy with the factorized structure of finite dimensional
 modules of the $q-$Onsager algebra considered in
\cite{ITaug,Bas3}.\vspace{1mm}

To conclude, let us mention that finite dimensional modules of the form (\ref{modeval}) have found applications  in the solution
 of the open XXZ spin chain with generic integrable boundary conditions \cite{BK2}.  For the family of tridiagonal pairs discussed in \cite{Bas3}, up to a normalization, the eigenvectors in each basis have been determined explicitly. In particular, they possess a factorized structure (see Proposition 3.3 and 3.7 of \cite{Bas3}), parallel to the factorized structure (\ref{AWgen}).\vspace{2mm}

\section{The $q-$Dolan-Grady hierarchy and spectral problem revisited}
As mentioned in the Introduction, there is a rather large class of quantum integrable models on the continuum or lattice 
whose local integrals of motion can be written in terms of the elements of the so-called $q-$Dolan-Grady hierarchy (\ref{qDGhier}),
an Abelian subalgebra of the $q-$Onsager algebra (\ref{qDG}).
For the integrable models that fall in this class, finding the spectrum and eigenstates of the Hamiltonian (\ref{Hamil}) relies on studying
the spectral problem of the elements $\{{\textsf I}_{2k+1}|k=1,2,...\}$ using the representation theory of the $q-$Onsager algebra.
For instance, in the literature, the two-dimensional Ising and 
superintegrable Potts models have been studied in details for $q=1$ using the explicit relation between the Onsager algebra and
a fixed-point subalgebra of $\widehat{sl_2}$ (under the action of a certain automorphism of $\widehat{sl_2}$ \cite{Davies}). For $q\neq 1$, the explicit relation between the $q-$Onsager algebra and a certain coideal subalgebra of $U_q(\widehat{sl_2})$ has been used to analyze the finite open XXZ 
spin chain \cite{BK2} or its thermodynamic limit analog  for various types of boundary conditions \cite{BB3,BKoj}. According to the size of the system - finite or semi-infinite -, finite or infinite dimensional representations ($q-$vertex operators) of the $q-$Onsager algebra have been considered.\vspace{1mm} 

In the previous Sections,   modules of the $q-$Onsager algebra - either infinite or finite dimensional - have been constructed in terms of Gasper-Rahman multivariable polynomials.
The purpose of this Section is to reformulate the spectral problem for the mutually commuting quantities 
$\{{\textsf I}_{2k+1}\}$ as a system of $q-$difference or difference equations (the case $q=1$ will be considered separately in Section 5). It implies that eigenfunctions of Hamiltonians of the form (\ref{Hamil}) can be expressed as combinations of Gasper-Rahman multivariable polynomials. As an application, a $q-$hypergeometric formulation of the eigenfunctions of the open XXZ chain with generic boundary conditions is given for generic parameters, see (\ref{eigenftwo}). Also, a new derivation of Nepomechie's relations is proposed and generalized, see (\ref{rel}).

\subsection{Infinite dimensional modules}
If we assume that the spectrum of ${\textsf I}_1$ is non-degenerate \cite{F15}, common eigenstates of ${\textsf I}_{2k+1}$ are uniquely determined from solving the spectral problem of ${\textsf I}_1$. For generic parameters $\omega_0,\omega_1,g_\pm$, we assume such hypothesis holds true. Then,  the action of this operator with respect to the basis of Gasper-Rahman polynomials is derived as follows. Denote:
\beqa
\Phi^{\pm} &=& \sum_{ \nu\in \{-1,0,1\}^N \backslash \nu_1=\pm 1  }\overline{C}_{\nu}(\{z\})|_{\nu_1=\pm 1} \oE_{z_2}^{\nu_2} \cdots \oE_{z_N}^{\nu_N},\nonumber\\
 \Phi^{0} &=& \sum_{ \nu\in \{-1,0,1\}^N \backslash \nu_1=0  }\overline{C}_{\nu}(\{z\})|_{\nu_1=0} \oE_{z_2}^{\nu_2} \cdots \oE_{z_N}^{\nu_N} + \frac{4\alpha_{N+1}}{\alpha_0(q^2+1)}x_0x_{N+1}.\nonumber
\eeqa
Using the results of previous Sections and the definition of ${\textsf{G}}_{1},\tilde{\textsf{G}}_{1}$ in terms of ${\textsf{W}}_{0},{\textsf{W}}_{1}$ given in (\ref{defel}), on the multivariable polynomial basis, the elements ${\textsf{W}}_{0},{\textsf{W}}_{1}$ of the $q-$Onsager algebra (\ref{qDG}) and their first `descendants' ${\textsf{G}}_{1},{\tilde{\textsf{G}}}_{1}$ act, respectively, as:
\beqa
\cW_0^{(N)}&=& \frac{1}{2}(z_1+z_1^{-1}),\label{ops}\\
\cW_1^{(N)}&=& \frac{\alpha_0 q}{2\alpha_{N+1}} \left(    \Phi^{+}\oE_{z_1}    +  \Phi^{-}\oE_{z_1}^{-1} +  \Phi^{0}   \right),\nonumber\\
\cG_1^{(N)}&=& \frac{\alpha_0 q}{4\alpha_{N+1}} \left(  (q^3-q^{-1})( z_1 \Phi^{+}\oE_{z_1}    + z_1^{-1} \Phi^{-}\oE_{z_1}^{-1}) +  (q-q^{-1})(z_1+z_1^{-1}) \Phi^{0}   \right),\nonumber\\
\tilde{\cG}_1^{(N)}&=& \frac{\alpha_0 q}{4\alpha_{N+1}} \left(  (q-q^{-3})( z_1^{-1} \Phi^{+}\oE_{z_1}    + z_1 \Phi^{-}\oE_{z_1}^{-1}) +  (q-q^{-1})(z_1+z_1^{-1}) \Phi^{0}   \right)\nonumber.
\eeqa
\vspace{1mm}

As a consequence, on the multivariable polynomial basis the first element ${\textsf  I}_{1}= \omega_0{\textsf W}_{0}+ \omega_1{\textsf W}_{1}+ g_+{\textsf G}_{1} + g_-{\tilde{\textsf G}}_{1}$ of the $q-$Dolan-Grady hierarchy (\ref{qDGhier}) acts as the multivariable  $q-$difference operator:
\beqa
{\cal I}^{(N)}_{1}&=&\frac{\alpha_0q}{2\alpha_{N+1}}\left( \Big(\omega_1 + \frac{(q^2-q^{-2})}{2}\big( g_+ qz_1 + g_- q^{-1}z_1^{-1}\big)   \Big)\Phi^+ \oE_{z_1}\right.\nonumber\\
&& \left. \qquad \quad + \   \Big(\omega_1 + \frac{(q^2-q^{-2})}{2}\big( g_+ qz_1^{-1} + g_- q^{-1}z_1   \big) \Big)\Phi^- \oE_{z_1}^{-1}\right)\nonumber\\
&&  \qquad \quad \ + \frac{\omega_0(z_1+z_1^{-1})}{2} + \frac{\alpha_0q}{2\alpha_{N+1}}\Big(
\omega_1 + \frac{1}{2}\big(g_+ + g_-\big)(q-q^{-1})(z_1+z_1^{-1})\Big)\Phi^0.\nonumber
\eeqa

\vspace{3mm}

Alternatively, in the multivariable polynomial basis, the above operators (\ref{ops}) can be written as semi-infinite block tridiagonal matrices, whose entries are explicitly expressed in terms of (\ref{bcoeff}), (\ref{ccoeff}) and (\ref{acoeff}). By straightforward calculations combining all expressions, on the multivariable polynomial $\hQ^{(N)}(\{n\},\{x\},\{\alpha\})$,
\beqa
{\cal I}^{(N)}_{1} \quad & \mbox{acts as}&\quad \!\!\!\!\!\!\!\!\! \sum_{ \{\nu_2,\nu_3,\dots,\nu_N\}\in \{-1,0,1\}^{N-1}  }\!\!\!\!\!\!\!  \left({\cal B}_{n_1 n_2 \cdots n_N}^{[\nu_N\ \nu_{N-1}-\nu_N\ \cdots 1-\nu_2]} E^{\nu_{N}}_{n_1}  E^{\nu_{N-1}-\nu_N}_{n_2}\cdots   E^{\nu_2-\nu_3}_{n_{N-1}}  E^{1-\nu_2}_{n_N}\right.  \nonumber\\
&& \qquad\qquad\qquad\qquad \quad \left.+ \ {\cal C}_{n_1 n_2 \cdots n_N}^{[\nu_N\ \nu_{N-1}-\nu_N\ \cdots -1-\nu_2]}  E^{\nu_{N}}_{n_1}  E^{\nu_{N-1}-\nu_N}_{n_2}\cdots   E^{\nu_2-\nu_3}_{n_{N-1}}  E^{-1-\nu_2}_{n_N}\right.\nonumber\\
&&  \qquad\qquad\qquad \qquad\quad \left.+ \  {\cal A}_{n_1 n_2 \cdots n_N}^{[\nu_N\ \nu_{N-1}-\nu_N\ \cdots -\nu_2]} E^{\nu_{N}}_{n_1}  E^{\nu_{N-1}-\nu_N}_{n_2}\cdots   E^{\nu_2-\nu_3}_{n_{N-1}}  E^{-\nu_2}_{n_N}\right),
 \nonumber
\eeqa
where
\beqa
\qquad {\cal B}_{n_1 n_2 \cdots n_N}^{[\nu_N\ \nu_{N-1}-\nu_N\ \cdots 1-\nu_2]} \!\!\! &=&\!\!\! \Big(\omega_0 + \frac{(q^2-q^{-2})}{2}\Big( g_-\frac{\alpha_0q}{\alpha_{N+1}}q^{-2\kN_N-1} + g_+\frac{\alpha_{N+1}}{\alpha_0q}q^{2\kN_N+1}\Big)  \Big) b_{n_1 n_2 \cdots n_N}^{[\nu_N\ \nu_{N-1}-\nu_N\ \cdots 1-\nu_2]} ,    \label{coefffinI}\\
{\cal C}_{n_1 n_2 \cdots n_N}^{[\nu_N\ \nu_{N-1}-\nu_N\ \cdots -1-\nu_2]}  \!\!\!&=&\!\!\!  \Big(\omega_0 + \frac{(q^2-q^{-2})}{2}\Big( g_-\frac{\alpha_{N+1}}{\alpha_0q}q^{2\kN_N-1} + g_+\frac{\alpha_0q}{\alpha_{N+1}}q^{-2\kN_N+1}\Big)   \Big) c_{n_1 n_2 \cdots n_N}^{[\nu_N\ \nu_{N-1}-\nu_N\ \cdots -1-\nu_2]},  \nonumber\\
{\cal A}_{n_1 n_2 \cdots n_N}^{[\nu_N\ \nu_{N-1}-\nu_N\ \cdots -\nu_2]} \!\!\! &=&\!\!\!    \omega_1  \theta^*_{\{n\}}\delta_{\{\nu\},\{0\}} + \Big(
\omega_0 + \big(g_+ + g_-\big)(q-q^{-1})   \theta^*_{\{n\}} \Big)  a_{n_1 n_2 \cdots n_N}^{[\nu_N\ \nu_{N-1}-\nu_N\ \cdots -\nu_2]}. \nonumber
\eeqa

For generic parameters $\omega_0,\omega_1,g_\pm$, the above coefficients are non-vanishing so that
there is no multivariable polynomial subspace left invariant under the action of the semi-infinite block tridiagonal matrix representing $\textsf{I}_1$. Corresponding eigenfunctions of the spectral problem:
\beqa
{\cal I}^{(N)}_{1} \Psi^{(N)}(\{x\},\{\alpha\}) =  \Lambda^{(N)}_1\Psi^{(N)}(\{x\},\{\alpha\}) 
\eeqa
are, in general, not polynomials. However, given the multivariable polynomial infinite dimensional basis of the vector space, they can be written as: 
\beqa
\Psi^{(N)}(\{x\},\{\alpha\})=\sum_{ \{n\}\in {I\!\!N}^N} f_{\{n\}} \hQ^{(N)}(\{n\},\{x\},\{\alpha\}),\label{eigenf}
\eeqa
where the coefficients $f_{\{n\}}$ satisfy a coupled system of three-term difference equations, with respect to the integer $\kN_N$. For generic parameters and $q$ not a root of unity, the above eigenfunctions diagonalize the full $q-$Dolan-Grady hierarchy. As a consequence of (\ref{Hamil}), the Hamiltonian is diagonalized by (\ref{eigenf}) too.

\begin{example} 
For $N=1$, recall $\{\alpha\}=\{\alpha_0,\alpha_1,\alpha_2,\alpha_3\}$. The eigenfunctions of the Hamiltonian   (\ref{Hamil}) are given by  $\Psi^{(1)}(x_1,\{\alpha\})=\sum_{n_1=0}^{\infty} f_{n_1} \hQ^{(1)}(n_1,x_1,\{\alpha\})$ where the coefficients $f_{n_1}$ satisfy the three-term recurrence relation ($f_{-1}=0,f_0=1$):
\beqa
f_{n_1-1}{\cal B}^{[1]}_{n_1-1}  + f_{n_1+1}{\cal C}^{[-1]}_{n_1+1} + f_{n_1}({\cal A}^{[0]}_{n_1} - \Lambda_1) = 0 \quad \mbox{for} \quad n_1=0,1,... \nonumber
\eeqa
\end{example}

\begin{example} 
For $N=2$, recall $\{\alpha\}=\{\alpha_0,\alpha_1,\alpha_2,\alpha_3,\alpha_4\}$. The eigenfunctions of the Hamiltonian   (\ref{Hamil}) are given by  $\Psi^{(2)}(x_1,x_2,\{\alpha\})=\sum_{n_1,n_2=0}^{\infty} f_{n_1,n_2} \hQ^{(1)}(n_1,n_2,x_1,x_2,\{\alpha\})$ where the coefficients $f_{n_1,n_2}$ satisfy the coupled recurrence relations ($f_{-1,n_2}=f_{n_1,-1}=f_{n_1+1,-2}=0$):
\beqa
&&f_{n_1-1,n_2}{\cal B}^{[10]}_{n_1-1,n_2}  + f_{n_1,n_2-1}{\cal B}^{[01]}_{n_1,n_2-1} + f_{n_1+1,n_2-2}{\cal B}^{[-12]}_{n_1+1,n_2-2}\nonumber\\
&&+ f_{n_1+1,n_2}{\cal C}^{[-10]}_{n_1+1,n_2} + f_{n_1,n_2+1}{\cal C}^{[0-1]}_{n_1,n_2+1} + f_{n_1-1,n_2+2}{\cal C}^{[1-2]}_{n_1-1,n_2+2} \nonumber\\
&&+ f_{n_1-1,n_2+1}{\cal A}^{[1-1]}_{n_1-1,n_2+1} + f_{n_1+1,n_2-1}{\cal A}^{[-11]}_{n_1+1,n_2-1} + f_{n_1,n_2}({\cal A}^{[00]}_{n_1,n_2}- \Lambda_1) = 0 \quad \mbox{for} \quad n_1,n_2=0,1,... \nonumber
\eeqa
\end{example}
\vspace{1mm}

Note that according to Theorem \ref{thmortho}, provided the parameters $\{\alpha\}$ satisfy the conditions (\ref{condpar}), scalar product between the eigenstates (\ref{eigenf}) can be computed in a straightforward manner and expressed in terms of $q-$shifted factorials (see (\ref{ortH})).

\subsection{Finite dimensional modules}
Let ${\cal V}$ be defined as in Proposition \ref{condpargen}. For generic parameters $\omega_0,\omega_1,g_\pm$, using the bispectral property, an eigenstate of ${\textsf I}_1$ (and more generally of $H$ given by (\ref{Hamil})) can be written in terms of Gasper-Rahman polynomials defined on a discrete support, either
\beqa
\qquad \Psi^{(N)}(\{\tilde{n}\},\{\alpha\})=\sum_{\kN_N=1}^{d_N} f_{\{n\}} \tfQ^{(N)}(\{n\},\{\tilde{n}\},\{\alpha\}) \quad \mbox{or} \quad \Psi^{(N)}(\{n\},\{\alpha\})=\sum_{\tilde{\kN}_N=1}^{d_N} f_{\{\tilde{n}\}} \tfQ^{(N)}(\{n\},\{\tilde{n}\},\{\alpha\}), \label{eigenftwo}
\eeqa
where the coefficients $f_{\{n\}}$ (resp. $f_{\{\tilde{n}\}}$) satisfy a coupled system of three-term difference equations, with respect to the integer $\kN_N$ (resp. $\tilde{\kN}_N$). The explicit form of these equations follows from the action of ${\textsf I}_1$ on the bispectral polynomial basis: the system of recurrence relations determining  $f_{\{n\}}$ is associated with the coefficients (\ref{coefffinI}), whereas the `dual' system of recurrence  relations determining  $f_{\{\tilde{n}\}}$ is associated with:
\beqa
 {\cal B}_{\tilde{n}_1 \tilde{n}_2 \cdots \tilde{n}_N}^{*[\nu_N\ \nu_{N-1}-\nu_N\ \cdots 1-\nu_2]}  \!\!\!&=&\!\!\!\Big(\omega_1 + \frac{(q^2-q^{-2})}{2}\Big( g_+\frac{\alpha_{N+1}\alpha_{N+2}}{\alpha_1}q^{2\tilde{\kN}_N+1} + g_- \frac{\alpha_1}{\alpha_{N+1}\alpha_{N+2}}q^{-2\tilde{\kN}_N-1}\Big)  \Big) {b}_{\tilde{n}_1 \tilde{n}_2 \cdots \tilde{n}_N}^{*[\nu_N\ \nu_{N-1}-\nu_N\ \cdots 1-\nu_2]},     \nonumber\\
{\cal C}_{\tilde{n}_1 \tilde{n}_2 \cdots \tilde{n}_N}^{*[\nu_N\ \nu_{N-1}-\nu_N\ \cdots -1-\nu_2]}   \!\!\!&=&  \!\!\!\Big(\omega_1 + \frac{(q^2-q^{-2})}{2}\Big( g_+   \frac{\alpha_1}{\alpha_{N+1}\alpha_{N+2}}  q^{-2\tilde{\kN}_N+1} + g_-       \frac{\alpha_{N+1}\alpha_{N+2}}{\alpha_1}q^{2\tilde{\kN}_N-1}\Big)  \Big)
 {c}_{\tilde{n}_1 \tilde{n}_2 \cdots \tilde{n}_N}^{*[\nu_N\ \nu_{N-1}-\nu_N\ \cdots -1-\nu_2]},  \nonumber\\
{\cal A}_{\tilde{n}_1 \tilde{n}_2 \cdots \tilde{n}_N}^{*[\nu_N\ \nu_{N-1}-\nu_N\ \cdots -\nu_2]}  &=&    \omega_0  \theta_{\{\tilde{n}\}}\delta_{\nu,0} + \Big(
\omega_1 + \big(g_+ + g_-\big)(q-q^{-1})   \theta_{\{\tilde{n}\}} \Big)  {a}_{\tilde{n}_1 \tilde{n}_2 \cdots \tilde{n}_N}^{*[\nu_N\ \nu_{N-1}-\nu_N\ \cdots -\nu_2]}. \nonumber
\eeqa
Here, the coefficients ${b}_{\tilde{n}_1  \cdots \tilde{n}_N}^{*[\nu_N\ \cdots 1-\nu_2]},{c}_{\tilde{n}_1  \cdots \tilde{n}_N}^{*[\nu_N\ \cdots -1-\nu_2]},{a}_{\tilde{n}_1  \cdots \tilde{n}_N}^{*[\nu_N\ \cdots -\nu_2]}$ follow from the transformation (\ref{dualD}) in (\ref{qdiff3}).
 
\begin{rem} Let ${\cal V}$ be the finite dimensional  module of the $q-$Onsager algebra, of dimension $2^N$, defined as in example \ref{modXXZ}.  Any of the states (\ref{eigenftwo}) provide an eigenstate of the Hamiltonian of the open XXZ chain with generic boundary conditions. 
\end{rem}

\subsection{Invariant subspaces and Nepomechie's relations}
For special relations among the parameters  $\omega_0,\omega_1,g_\pm$, multivariable polynomial subspaces that are left invariant under the action of the element $\textsf{I}_1$ can be constructed. Let ${\cal V}$ be defined  as in Proposition \ref{condpargen}. In particular, $\mbox{max}(\kN_N)=\mbox{max}(\tilde{\kN}_N)=d_N$.
Let $P,P^*$ be integers. Define
\beqa
{\cal W}_+&=&\mathrm{Span}\{\tfQ^{(N)}(\{n\},\{\tilde{n}\},\{\alpha\})|0\leq \kN_N\leq P\},\quad \{\tilde{n}\} \quad \mbox{fixed},\nonumber\\
{\cal W}_-&=&\mathrm{Span}\{\tfQ^{(N)}(\{n\},\{\tilde{n}\},\{\alpha\})|0\leq \tilde{\kN}_N\leq P^*\},\quad \{n\} \quad \mbox{fixed},\nonumber\\
\overline{\cal W}_+&=&\mathrm{Span}\{\tfQ^{(N)}(\{n\},\{\tilde{n}\},\{\alpha\})|P+1\leq \kN_N\leq \mbox{max}(\kN_N)\},\quad \{\tilde{n}\} \quad \mbox{fixed},\nonumber\\
\overline{\cal W}_-&=&\mathrm{Span}\{\tfQ^{(N)}(\{n\},\{\tilde{n}\},\{\alpha\})|P^*+1\leq \tilde{\kN}_N\leq \mbox{max}(\tilde{\kN}_N)\}, \quad \{n\} \quad \mbox{fixed}.\nonumber
\eeqa
According to the structure of the coefficients of the block tridiagonal matrices ${\cal I}^{(N)}_{1}$ associated with $\textsf{I}_1$, it follows:
\beqa
&&{\cal I}^{(N)}_{1} {\cal W}_+ \subset  {\cal W}_+   \qquad  \mbox{if}\qquad  \omega_0 + \frac{(q^2-q^{-2})}{2}\Big( g_-\frac{\alpha_0q}{\alpha_{N+1}}q^{-2P-1} + g_+\frac{\alpha_{N+1}}{\alpha_0q}q^{2P+1}\Big) =0 , \label{rel}\\
&&{\cal I}^{(N)}_{1} {\cal W}_- \subset  {\cal W}_-   \qquad  \mbox{if}\qquad  \omega_1 + \frac{(q^2-q^{-2})}{2}\Big( g_+\frac{\alpha_{N+1}\alpha_{N+2}}{\alpha_1}q^{2P^*+1} + g_- \frac{\alpha_1}{\alpha_{N+1}\alpha_{N+2}}q^{-2P^*_N-1}\Big)=0,\nonumber\\
&&{\cal I}^{(N)}_{1} \overline{\cal W}_+ \subset  \overline{\cal W}_+   \qquad  \mbox{if}\qquad \omega_0 + \frac{(q^2-q^{-2})}{2}\Big( g_-\frac{\alpha_{N+1}}{\alpha_0q}q^{2P+1} + g_+\frac{\alpha_0q}{\alpha_{N+1}}q^{-2P-1}\Big)=0 , \nonumber\\
&&{\cal I}^{(N)}_{1} \overline{\cal W}_- \subset  \overline{\cal W}_-   \qquad  \mbox{if}\qquad \omega_1 + \frac{(q^2-q^{-2})}{2}\Big( g_+   \frac{\alpha_1}{\alpha_{N+1}\alpha_{N+2}}  q^{-2P^*-1} + g_-       \frac{\alpha_{N+1}\alpha_{N+2}}{\alpha_1}q^{2P^*+1}\Big)=0 .\nonumber
\eeqa
\vspace{1mm}

On each of the invariant subspaces, eigenstates of the Hamiltonian (\ref{Hamil}) can be systematically constructed as truncations of any of the two possible combinations (\ref{eigenftwo}). For explicit examples, see \cite[Section~3.2]{BK2}.\vspace{1mm}

In the literature on quantum integrable systems with boundaries, special cases of the above relations among the parameters in (\ref{rel})  first appeared in the algebraic Bethe ansatz analysis of the open XXZ spin chain with non-diagonal boundary conditions. 
Sometimes refered as `Nepomechie's relations' \cite[eq.~(1.4)]{N}, they usually characterize the family of constraints between the left and right integrable boundary conditions of the system for which the standard Bethe ansatz approach can be applied. Within the $q-$Onsager approach of the open XXZ spin chain, they were independantly identified from a representation theory perspective in \cite{BK2}. Importantly, it provided a rigorous proof of completeness of the spectrum of the Hamiltonian. Note that the above relations also occur in other approaches, see for instance \cite{FaKN}, as well as they characterize a class of solutions to the boundary $q-$Knizhnik-Zamolodchikov equations \cite{SV}. Here, we would like to stress that the relations (\ref{rel}) generalize the results of \cite{BK2}: they characterize the existence of a class of invariant subspaces of the Abelian subalgebra of the $q-$Onsager algebra. As a consequence, they should also arise for integrable models associated with higher spin representations, for instance.

\section{The special case $q=1$}
In this Section, new infinite and finite dimensional modules of the Onsager algebra \cite{Ons} with defining relations (\ref{qDG}) for $q=1$ are constructed in terms of the multivariable Krawtchouk polynomials introduced by Tratnik \cite{Trat} (see also \cite{GerI}).  The analysis is extended in a straightforward manner to more general tridiagonal algebras with defining relations (\ref{TD1}), (\ref{TD2}) for $\beta=2$. Restricting the parameters allows to derive explicit examples of tridiagonal pairs of Racah types (Case II), as briefly discussed.\vspace{1mm} 

First, we recall the duality and  bispectrality of the multivariable Krawtchouk polynomials \cite{Trat,GerI}. Although the notations may slightly differ, the basic  material is taken from \cite{GerI} to which the reader is refered for more details. In terms of the Gaussian hypergeometric function, the one-variable Krawtchouk polynomial is defined as:
\beqa
k_n(x,a,b)=(-b)_n\, {}_2\mathrm{F}_1\left[\begin{array}{c} -n,-x\\-b
\end{array};\frac{1}{a}\right],\ n\in {\mathbb N} \mbox{ and } (x,a,b)
\in{\mathbb R}^3, \nonumber
\eeqa
with $(y)_n=\Gamma(y+n)/\Gamma(y)$ a function which can be well defined
for all $y$ real and $n$ integer. Following \cite{Trat,GerI}, $N-$variable generalizations of the Krawtchouk
polynomial can be introduced. 

\begin{defn}[See \cite{GerI}]  Let $\alpha_1,...,\alpha_{N}$ and $M$ be nonzero real parameters.
 The normalized $N-$variable generalization of  the Krawtchouk polynomial is defined by:
\beqa
\widehat{K}^{(N)}(\{n\},\{x\},\{\alpha\};M)=\frac{1}{(-M)_{\kN_N}}\prod_{j=1}^N
k_{n_j}(x_j;\frac{\alpha_j}{1-\kA_{j-1}};M-\kN_N+\kN_j-\kX_{j-1}),
\label{mKrawt} 
\eeqa
where  the notation
\beqa
 \kX_j=x_1+x_2+\cdots+x_j,\quad \kX_0=0, \qquad
 \kA_j=\alpha_1+\alpha_2+\cdots+\alpha_j, \quad \kA_0=0 \nonumber
\eeqa
 is used.
\end{defn}
\vspace{2mm}

By analogy with the $q-$deformed case, an  inner product on the space of polynomials ${\cal P}$ can be introduced.  For $f,g\in{\cal P}$, it takes the form:
\beqa
\langle f,g\rangle = \sum_{\{x\}\in{\cal F}} f(x)g(x)\rho(x).\label{inner2}
\eeqa
With respect to this inner product, for a positive integer $\alpha_0$ and $\mathcal{A}_N<1$ the multivariable Krawtchouk polynomials (\ref{mKrawt}) are orthogonal on the domain ${\cal F}=\{ \{x\}\in{\mathbb N}^N| \kX_N\leq M\}$  for the weight $\rho(x)$ given by:
\beqa
\rho(x)=\frac{1}{(N-\kX_N)!}\prod_{k=1}^N\frac{1}{x_k!}\left(
\frac{\alpha_k}{1-\kA_N}\right)^{x_k}.\nonumber
\eeqa
\begin{rem}  The normalized multivariable Meixner polynomials are obtained from (\ref{mKrawt}) through the substitution:
$$\alpha_k=\frac{c_k}{c_1+\cdots+c_N-1} \quad \mbox{and} \quad M=-s.$$
\end{rem}

\vspace{2mm}

Following \cite{GerI}, multivariable Krawtchouk polynomials are known to diagonalize two different commutative algebras of difference operators in the variables $\{x\}$ and $\{n\}$, respectively. Let 
\beqa
{\cal D}'_{x}={\mathbb R}(x_1,x_2,...,x_N)[E^{\pm 1}_{x_1},E^{\pm 1}_{x_2},...,E^{\pm 1}_{x_N}]\nonumber
\eeqa
denote the associative algebra of difference operators with rational functions of $x_1,x_2,...,x_N$ as coefficients.
A mutually commuting family of difference operators which are diagonalized by the multivariable Krawtchouk polynomials (\ref{mKrawt}) has been constructed in \cite{GerI}. Here, they are denoted $\{{d\!I}^{*(l)}_{\{x\}}|l=1,...,N\}$. 

%
\begin{defn}[See \cite{GerI}, Section 5.4] The $N$-variable  Krawtchouk difference operators $\{{d\!I}^{*(l)}_{\{x\}}|l=1,2,...,N\}$ are defined as:
 \beqa
\qquad&& {d\!I}^{*(l)}_{\{x\}} \equiv {\cal L}_l\left(x_{N+1-l},x_{N+2-l},...,x_N,\frac{\alpha_{N+1-l}}{1-\kA_{N-l}},\frac{\alpha_{N+2-l}}{1-\kA_{N-l}},...,\frac{\alpha_{N}}{1-\kA_{N-l}};M-\kX_{N-l}\right)\label{diffq1}
\eeqa
where
\beqa
{\cal L}_N(x_1,x_2,...,x_N,\alpha_1,\alpha_2,....,\alpha_N;M)&=&(\kX_N-M)\sum_{i=1}^N\alpha_iE_{x_i}+(\kA_N-1)
\sum_{i=1}^Nx_iE_{x_i}^{-1}\nonumber\\
&&\qquad -\!\!\!\sum_{1<i\neq j<N}\!\!\!\alpha_i x_j E_{x_i}E_{x_j}^{-1}-\sum_{i=1}^N\alpha_i x_i + \kX_N+\kA_N
(M-\kX_N).\nonumber
\eeqa
\end{defn}
\vspace{2mm}

The algebra generated by these difference operators can be understood as a limiting case of the commutative algebra ${\cal A}_z$ of the $q-$difference operators considered in Section 2. According to \cite{GerI}, the difference operators  ${d\!I}^{*(l)}_{\{x\}}|l=1,2,...,N\}$ are mutually commuting:
\beqa
\big[ {d\!I}^{*(l)}_{\{x\}} , {d\!I}^{*(j)}_{\{x\}} \big]=0 \quad \mbox{for all}\quad l,j\in\{1,2,...,N\}.
\eeqa

As explained in \cite{GerI}, the normalized multivariable Krawtchouk polynomials (\ref{mKrawt}) are invariant under the following change of variables
$$x_j \rightarrow \tilde{n}_{N+1-j},\quad \ n_j \rightarrow
\tilde{x}_{N+1-j},\ \quad \alpha_j  \rightarrow  \frac{\alpha_{N+1-j}(1-\mathcal{A}_N)}{(1-\mathcal{A}_{N+1-j})
(1-\mathcal{A}_{N-j})},\ \quad M \rightarrow  M.$$
In particular, it implies the duality relation \cite[Theorem~5.6]{GerI}: 
\beqa
\widehat{K}^{(N)}(\{n\},\{x\},\{\alpha\};M)=
\widehat{K}^{(N)}(\{\tilde{n}\},\{\tilde{x}\},\{\tilde{\alpha}\};M) 
\eeqa

By analogy with the case $q\neq 1$, an isomorphism ${\mathfrak b}_1$ between the associative algebras of difference operators in the variable $\{x\}$ and $\{n\}$ that extends the above map can be introduced:
\beqa 
&&{\mathfrak b}_1(M)=M, \qquad {\mathfrak b}_1(\alpha_j)=\frac{\alpha_{N+1-j}(1-
\mathcal{A}_N)}{(1-\mathcal{A}_{N+1-j})(1-\mathcal{A}_{N-j})}\qquad \mbox{for}\ j=1,2,...,N,\nonumber\\
&&{\mathfrak b}_1(x_j) =  n_{N+1-j},\qquad {\mathfrak b}_1(E_{x_j}) =  E_{n_{N+1-j}}\qquad\qquad \qquad \ \ \mbox { for } j=1,2,...,N.\nonumber
\eeqa
Then, a `dual' commutative algebra generated by a family of difference operators $\{{d\!I}^{(l)}_{\{n\}}|l=1,...,N\}$, follows from (\ref{diffq1}) through the action of ${\mathfrak b}_1$.  According to the definitions of \cite{GerI}, one has:
\begin{defn}[See \cite{GerI}]
Let $l=1,2,...,N$. The  dual mutually commuting $N$-variable Krawtchouk difference operators are defined by: 
\beqa
 {d\!I}_{\{n\}}^{(l)}= {\kb}_1({d\!I}_{\{x\}}^{*(l)}).
\eeqa
\end{defn}

\begin{example} For $k=N$, the dual $N$-variable Krawtchouk difference operator takes the form:
\beqa
{d\!I}^{(N)}_{\{n\}} & = &(\kN_N-M) \sum_{i=1}^N{\mathfrak b}
(\alpha_{N+1-i})E_{n_i}+{\mathfrak b}_1(\mathcal{A}_N-1)
\sum_{i=1}^Nn_{i}E_{n_i}^{-1}  \label{diffq3}\\ &&    -\sum_{1<i\not= j<N}{\mathfrak b}_1(\alpha_{N+1-i}) n_{j} 
E_{n_i}E_{n_j}^{-1}+ \kN_N+
{\mathfrak b}_1(\mathcal{A}_N)(M- \kN_N)-\sum_{i=1}^N
{\mathfrak b}_1(\alpha_{N+1-i})n_{i}.\nonumber
\eeqa
\end{example}
\vspace{2mm}

As a consequence of the duality relation, it follows \cite{GerI}:
\begin{thm}[See \cite{GerI}, Theorem 5.7]\label{bispec3} Let
  $(\{\alpha\},M)\in  {\mathbb R}^{N+1}$ and $l=1,2,...,N$. The normalized multivariable
  polynomial $\widehat{K}^{(N)}(\{n\},\{x\},\{\alpha\};M)$ solves the following system
  of difference-difference bispectral problems: \begin{eqnarray}
 {d\!I}^{*(l)}_{\{x\}} \widehat{K}^{(N)}(\{n\},\{x\},\{\alpha\};M) & = & (\kN_N- \kN_{N-l})
\widehat{K}^{(N)}(\{n\},\{x\},\{\alpha\};M),\label{bispec11}\\
 {d\!I}^{(l)}_{\{n\}} \widehat{K}^{(N)}(\{n\},\{x\},\{\alpha\};M) & = & \kX_{N+1-l} 
\widehat{K}^{(N)}(\{n\},\{x\},\{\alpha\};M).\label{bispec12}
\eeqa
\end{thm}
\vspace{2mm}

For $N=1$, the equations (\ref{bispec11}), (\ref{bispec12}) produce the well-known difference equation and three-term recurrence relation of the one-variable Krawtchouk polynomials \cite{KS}, respectively. In particular, in the basis generated by the Krawtchouk polynomials, the operator $ {d\!I}^{(1)}_{\{n\}}$ defines a semi-infinite tridiagonal matrix. For $N$ generic, in the basis of multivariable Krawtchouk polynomials (\ref{mKrawt}) the structure of the matrix associated with ${d\!I}^{(N)}_{\{n\}}$ is determined as follows. Let $p\in {\mathbb N}$ be fixed. Consider the subspace $V_p$ generated by the polynomials $\{\widehat{K}^{(N)}(\{n\},\{x\},\{\alpha\};M)|\ \kN_N=p\}$, and act on it with (\ref{diffq3}). The action of the operators $E_{n_i}$, $E_{n_i}^{-1}$ and $E_{n_i}E_{n_j}^{-1}$ with $i\neq j$ in (\ref{diffq3})  on $V_p$ sends it to $V_{p+1}$,  $V_{p-1}$ and $V_{p}$, respectively. As a consequence, 
\beqa
{d\!I}^{(N)}_{\{n\}} V_p \subseteq V_{p+1} + V_{p} + V_{p-1}.\label{tri1}
\eeqa
For $N$ generic, we conclude that the matrix associated with ${d\!I}^{(N)}_{\{n\}}$ is semi-infinite block tridiagonal. For $N=1,2$, note that the explicit expressions of ${d\!I}^{(l)}_{\{n\}}$  for $N=1$ $(l=1)$ and $N=2$ $(l=1,2)$ can be obtained from \cite[Appendix~A]{GerI}. Above arguments obviously extend to any operator  ${d\!I}^{(l)}_{\{n\}}$ with $l<N$.
\vspace{2mm}

\subsection{Infinite and finite dimensional modules of the Onsager algebra}
The defining relations of the Onsager algebra are given by (\ref{qDG}) with $q=1$. Let us endow the vector space ${\mathbb R}[x_1,x_2,...,x_N]$ with a module structure of the Onsager algebra. By analogy with the analysis in Section 2, a family of homomorphisms indexed by the integer $l=1,2,...,N$ can be exhibited as follows.
\begin{prop}\label{realDG} The map defined by 
\beqa
\textsf{W}_0  \mapsto   \kX_{N+1-l}, \qquad \textsf{W}_1  \mapsto {d\!I}^{*(l)}_{\{x\}} , \qquad \rho \mapsto 1 , \label{actWgenOA}
\eeqa
is an homomorphism from $O_1(\widehat{sl_2})$ to ${\cal D}'_x$.
\end{prop}
\begin{proof} To show the claim for $l=N$, one essentially follows the same steps as in Proposition \ref{realqDG2}. The proof uses the fact that
\beqa
x_1^2 -2x_1E^{\pm 1}_{x_1}(x_1) + (E^{\pm 1}_{x_1}(x_1))^2 -1=0,\\
(\kN_N\pm 1)^2 - 2(\kN_N\pm 1)\kN_N + \kN_N^2 -1 =0,\nonumber
\eeqa 
together with the observation that the matrix representating $\textsf{W}_0$ is semi-infinite block tridiagonal (it acts as (\ref{tri1})) in the basis which diagonalizes  $\textsf{W}_1$. Similar arguments are used to show the claim for $1\leq l < N$. 
\end{proof}

It follows: 
\begin{thm}
 Let $V$ denote the infinite dimensional vector space ${\mathbb R}[ x_1,x_2,...,x_N]$.  Let $\textsf{W}_0,\textsf{W}_1$ be the standard generators of the Onsager algebra (\ref{qDG}) with $q=1$, where $\textsf{W}_0,\textsf{W}_1$ act on $V$ as (\ref{actWgenOA}). Then, $\textsf{W}_0,\textsf{W}_1$ induce on $V$ a module structure for the Onsager algebra with $\rho=1$.
The multivariable Krawtchouk polynomials $\widehat{K}^{(N)}(\{n\},\{x\},\{\alpha\};M)|\{n\}\in ({\mathbb N})^N; M\in {\mathbb R}\}$  form a basis of this module. 
\end{thm}

Let us finally construct finite dimensional  modules of the Onsager algebra in terms of multivariable Krawtchouk polynomials. On one hand, according to the expression (\ref{diffq1}), if $x_j=0$ observe that the contributions from the downward shift operator $E^{-1}_{x_j}$ disappears. On the other hand, let $J\in {\mathbb N}$. Suppose $x_i=J$ for a given $i$. Then, it is possible to cancel all contributions associated with the  upward shift operators in the first and third terms of  (\ref{diffq1}) as follows: 
we require both that $\kX_N\leq J$ and $M=J$. Indeed, with these two conditions, note in particular that when $x_i=J$, $\kX_N=J$ and all the other $x_j$ with $j\neq i$ are necessarily $0$ which implies that the first and third term  of  (\ref{diffq1})   disappear.
Because of the simple duality relation between $x_i,E_{x_i}$ and
$n_i,E_{n_i}$, these conditions clearly also restrict the $n_i$ and
$\kN_N$ in the same way. It follows:
\begin{prop}\label{condpargen1}
 Let $J$ be a positive integer, and assume $M=J$. Define
\beqa
\mathcal{V}_J^{(N)} & = & \mathrm{Span}\{\widehat{K}^{(N)}(\{n\},\{x\},\{\alpha\};J)| 
\kX_N\leq J \} \quad \mbox{for}\quad  \kN_N\leq J  \quad \mbox{fixed},\nonumber\\
\mathcal{V}_J^{*(N)} & = & \mathrm{Span}\{\widehat{K}^{(N)}(\{n\},\{x\},\{\alpha\};J)| 
\kN_N\leq J \} \quad \mbox{for}\quad  \kX_N\leq J  \quad \mbox{fixed}.
\eeqa
Then $\mathcal{V}_J^{(N)},\mathcal{V}_J^{*(N)}$ are  finite dimensional
modules of the Onsager algebra.
\end{prop}

Note that the corresponding tridiagonal pair  has a reduced parameter sequence: the matrices associated with $\textsf{W}_0,\textsf{W}_1$  satisfy the defining relations of the Onsager algebra (i.e. (\ref{qDG}) with $q=1,\rho=1$).  The TD pair is  of Racah type (Case II), and the spectra have a simple structure with $a=a^*=c=c^*=0$ and $b=b^*=1$.

\subsection{Generalization to tridiagonal algebras with $\beta=2$}
 By analogy with the analysis of the previous Sections, it is straightforward to show that  the multivariable Racah polynomials of Tratnik \cite{Trat} generate infinite dimensional  modules of the tridiagonal algebras with defining relations (\ref{TD1}), (\ref{TD2}) for the special case $\beta=2$. Recall that the one-variable Racah polynomials are defined as \cite{F16}:
\beqa
r_n(x;a,b,c,d)=(a+1)_n(b+d+1)_n(c+1)_n  \  {}_4\mathrm{F}_3\left[\begin{array}{c} -n,n+a+b+1, -x,x+c+d+1\\ a+1,b+d+1,c+1
\end{array}; 1\right].\nonumber
\eeqa
Let $\zeta_0,\zeta_1,...,\zeta_{N+1}$ be real parameters. Let $N,M$ be positive integers. The normalized multivariable Racah polynomials are given by \cite{Trat} (see also \cite{GerI}):
\beqa
\widehat{R}_N(\{n\},\{x\},\{\zeta\},M)=\frac{R_N(\{n\},\{x\},\{\zeta\},M)}{(-M)_{\kN_N}(-M-\zeta_0)_{\kN_N}\prod_{k=1}^N(\zeta_{k+1}-\zeta_k)_{n_k}}\label{mRac}
\eeqa
where
\beqa
R_N(\{n\},\{x\},\{\zeta\},M)&=&\prod_{k=1}^N r_{n_k}(-\kN_{k-1}+x_k;2\kN_{k-1}+\zeta_k-\zeta_0-1,\zeta_{k+1}-\zeta_k-1,\nonumber\\
&&\qquad \qquad \kN_{k-1}-x_{k+1}-1,\kN_{k-1}+\zeta_k+x_{k+1})\nonumber
\eeqa
with $\kN_N\leq M$. With respect to the inner product (\ref{inner2}), the polynomials $\widehat{R}_N(\{n\},\{x\},\{\zeta\},M)$ are orthogonal on 
\beqa
{\cal F}=\{\{x\}\in {\mathbb N}^N|0\leq x_1\leq x_2\leq \cdots \leq x_N\leq M\}.
\eeqa
The explicit expression of the weight $\rho(x)$ follows from \cite[eq.~(3.8)]{GerI}. In the parametrization above, it reads ($x_0=0,x_{N+1}=M$):
\beqa
\rho(x)=\prod_{k=0}^N \frac{\Gamma(\zeta_{k+1}-\zeta_k +x_{k+1}-x_k)\Gamma(\zeta_{k+1}+x_{k+1}+x_k)}{(x_{k+1}-x_k)!\Gamma(\zeta_k+1+x_{k+1}+x_k)}\prod_{k=1}^{N}(\zeta_k+2x_k).\nonumber
\eeqa

\vspace{1mm}

In \cite{GerI}, using the duality property of the multivariable polynomials (\ref{mRac}), two commutative families of difference operators have been introduced. They are diagonalized by Tratnik's multivariable Racah polynomials. Using the notations in \cite{GerI}, let us denote:
\beqa
 \tilde{d\!I}^{*(l)}_{\{x\}} \equiv {\mathfrak L}^x_{l}\quad \mbox{and}  \quad \tilde{d\!I}^{(l)}_{\{n\}} \equiv  {\mathfrak L}^n_{l}  -M(M+\zeta_{N+1-l}) \quad \mbox{for} \quad l=1,2,...,N,
\eeqa
where ${\mathfrak L}^x_{l}$ and  ${\mathfrak L}^n_{l}$ are defined, respectively, in \cite[eqs. (3.24),(2.6) or (3.4)]{GerI} and  \cite[eq. (4.11)]{GerI}. By analogy with Theorem \ref{bispec}, it was shown:
\begin{thm}[See \cite{GerI}, Theorem 4.6]\label{bispec4} Let
 $(\{\zeta\},M)\in  {\mathbb R}^{N+3}$ and $l=1,2,...,N$. The normalized multivariable
  polynomial $\widehat{R}_N(\{n\},\{x\},\{\zeta\},M)$ solve the following system
  of difference-difference bispectral problems:
 \begin{eqnarray}
 \tilde{d\!I}^{*(l)}_{\{x\}}  \widehat{R}_N(\{n\},\{x\},\{\zeta\},M) & = & \lambda_{\{n\}}^{*(l)}
\widehat{R}_N(\{n\},\{x\},\{\zeta\},M),\label{bispecprime11}\\
 \tilde{d\!I}^{(l)}_{\{n\}} \widehat{R}_N(\{n\},\{x\},\{\zeta\},M) & = & \lambda_{\{x\}}^{(l)} 
\widehat{R}_N(\{n\},\{x\},\{\zeta\},M).\label{bispecprime12}
\eeqa
\end{thm}
where
\beqa
\lambda_{\{n\}}^{*(l)}&=& - (\zeta_{l+1}-\zeta_0)\kN_{l}-\kN_{l}(\kN_{l}-1),\label{specTD1}\\
\lambda_{\{x\}}^{(l)}&=&    - (1+\zeta_{N+1-l})x_{N+1-l} - x_{N+1-l}(x_{N+1-l}-1).\label{specTD2}
\eeqa
\vspace{2mm}

Infinite dimensional   modules of tridiagonal algebras can be constructed as follows. We endow the vector space ${\mathbb R}[x_1,x_2,...,x_N]$ with a module structure of a tridiagonal algebra $T(2,\gamma,\gamma^*,\rho,\rho^*)$, as in Definition \ref{TDgendef}. A family of homomorphisms indexed by the integer $l=1,2,...,N$ can be exhibited as follows.
\begin{prop}\label{realtridiag} The map defined by 
\beqa
\textsf{A}  \mapsto \lambda_{\{x\}}^{(l)}  , \quad \textsf{A}^*  \mapsto \tilde{d\!I}^{*(l)}_{\{x\}},&& \gamma\mapsto -2,
\quad \gamma^*\mapsto -2, \label{defAA*}\\
&& \rho \mapsto   \zeta^2_{N+1-l}-1, \quad  \rho^* \mapsto    (\zeta_0-\zeta_{l+1}+1)^2-1.  \nonumber
\eeqa
is an homomorphism from $T(2,\gamma,\gamma^*,\rho,\rho^*)$ to ${\cal D}'_x$.
\end{prop}
\begin{proof} 
To show the claim for $l=N$, one essentially follows the same steps as in Proposition \ref{realqDG2}. The main difference here is the structure of the two-variable polynomials associated with the l.h.s of (\ref{TD1}), (\ref{TD2}), respectively. Similarly to (\ref{po1}), (\ref{eq}), these polynomials are vanishing due to the structure of the spectra (\ref{specTD1}),  (\ref{specTD2}), provided the parameters $\gamma,\gamma^*,\rho,\rho^*$ are fixed to the appropriate values. Similar arguments are used to show the claim for $1\leq l < N$. 
\end{proof}
 
It follows: 
\begin{thm}
 Let $V$ denote the infinite dimensional vector space ${\mathbb R}[ x_1,x_2,...,x_N]$.  Let $\textsf{A},\textsf{A}^*$ be the standard generators of the tridiagonal algebra $T(2,-2,-2,\rho,\rho^*)$, where  $\textsf{A},\textsf{A}^*$ act on $V$ as (\ref{defAA*}). Then, $\textsf{A},\textsf{A}^*$  induce on $V$ a module structure for $T(2,-2,-2, \zeta^2_{N+1-l}-1, (\zeta_0-\zeta_{l+1}+1)^2-1)$.
The multivariable Racah polynomials \\ $\widehat{R}^{(N)}(\{n\},\{x\},\{\zeta\},M)|\{n\}\in ({\mathbb N})^N; M\in {\mathbb R}\}$  form a basis of this module. 
\end{thm}

To conclude this Section, let us observe that finite dimensional   modules of the tridiagonal algebra $T(2,\gamma,\gamma^*,\rho,\rho^*)$ based on the multivariable Racah polynomials (\ref{mRac}) can be easily derived by considering certain quotients of the space of parameters $\{\zeta\},M$. To this end, it is sufficient to identify the set of constraints on the parameters -  by analogy with Proposition \ref{condpargen} - using the explicit expressions for  ${\mathfrak L}^x_{l}$ and  ${\mathfrak L}^n_{l}$ taken from \cite[eq. (3.24)]{GerI} and  \cite[eq. (4.11)]{GerI}, respectively. The analysis being straighforward, we skip the details. Note that in the basis generated by the corresponding multivariable Racah polynomials (\ref{mRac}) defined on a discrete support, the elements  $\textsf{A}$, $\textsf{A}^*$ act as a tridiagonal pair of Racah type with spectra of the form (\ref{specTD1}), (\ref{specTD2}). This is in perfect agreement with the expected structure of the spectra that follows from the theory of tridiagonal pairs of Racah type (Case II).
\vspace{1mm}

\section{Perspectives}

Besides its elegance, for the class of quantum integrable models generated from the $q-$Onsager algebra, to have a $q-$hypergeometric formulation 
should be of practical importance. It should give a fresh look on existing solutions following from different frameworks such as the algebraic Bethe ansatz or separation of variables approaches but not only. We think it could pave the way to circumvent or reconsider some of the difficulties inherent to the structure of already existing analytical solutions.  Below, we list a conjecture and some open problems. \vspace{1mm}

{\it Irreducibility criteria:} We conjecture that the modules built on
Gasper-Rahman polynomials are irreducible for generic parameters
$\alpha_0,\alpha_1,...,\alpha_{N+2}$ and $q$, meaning that they do not
take specific values such as those of Proposition \ref{condpargen}.  \vspace{1mm} 

{\it Open XXZ spin chain:} For $q\neq 1$, one of the main motivation for constructing a $q-$hypergeometric basis associated with the $q-$Onsager algebra originates in the problem of solving quantum integrable systems 
with generic integrable boundary conditions, the XXZ open spin chain to cite the simplest example. On one hand, recall that for certain regimes 
of boundary parameters or generic parameters, through the Bethe ansatz \cite{N,CLSW,YZ,FiK},
 its modified versions \cite{Be,BeP,Cr,ZLCYSW}, functional alternatives \cite{Ga1,Ga2,CYSW} or the separation of variables approach \cite{FaKN,KMN}, 
eigenstates and eigenvalues are expressed in terms of (Bethe) roots of highly transcendental Bethe equations. It appears that
an explicit characterization of the Hamiltonian's eigenfunctions as polynomials offers the possibility of studying the correspondence between Bethe roots 
and solutions of algebraic equations \cite{F17}. On the other hand, in the thermodynamic limit, either from the algebraic Bethe ansatz setting \cite{KKMNST} or the $q-$vertex operator approach 
\cite{JKKKM,BKoj}, correlation functions and form factors are always given in terms of multiple integral representations. Although formally satisfying,  these are difficult objects to handle for practical applications. From the point of view developed here, the $q-$hypergeometric formulation (\ref{eigenftwo}) of the eigenfunctions opens the way to reconsidering scalar products of eigenstates in terms of Gasper-Rahman multivariable polynomials \cite{GR1}. Whether corresponding integral representations could be related with the expressions derived within the $q-$vertex operator approach or algebraic Bethe ansatz is an interesting  question.\vspace{1mm} 

{\it Higher rank  $q-$Onsager algebras and orthogonal polynomials}. Higher rank generalizations of the $q-$Onsager algebra (\ref{qDG}) have been introduced in the literature \cite[Definition~2.1]{BB1}. They can be realized in terms of certain coideal subalgebras of $U_q(\widehat{g})$ for any Kac-Moody algebra ${\widehat g}$ \cite[Proposition~2.1]{BB1} (see also \cite{Kolb}). For the simplest example $q=1$ and ${\widehat g}=A^{(1)}_{n-1}$, such a generalization is called the $sl_n-$Onsager algebra,  first introduced in \cite{Uglov}. In view of the results presented here, an interesting problem would be to identify the family of multivariable  orthogonal  polynomials that would be higher rank generalizations of those of Tratnik ($q=1$) or Gasper-Rahman ($q\neq 1$). Indeed, corresponding  infinite and finite dimensional modules of the generalized $q-$Onsager algebras will provide an explicit $q-$hypergeometric basis for a class of quantum integrable models, among which one finds open spin chains with higher rank symmetries. \vspace{1mm} 

{\it Non-polynomial solutions and $q-$Virasoro:} As an extension of Bochner's original theorem, the bispectral problem for tridiagonal doubly infinite matrices has been considered in the literature. In this case, solutions are intimately related to the Virasoro algebra but are no longer necessarily polynomials \cite{GH1}. In the context of $q-$special functions and the $q-$Virasoro algebra, a similar extension has been considered too \cite{GH2,GH3}. In view of these, it is 
clear that an extension of the bispectral problem (\ref{qdiffgen}),(\ref{recgen}) for $N>1$ can be considered similarly, see for instance \cite{GerI2}. These works may establish a relation between  the $q-$Onsager algebra and an algebraic structure extending the $q-$Virasoro algebra.\vspace{1mm}

Some of these problems will be considered elsewhere.

\vspace{5mm}

\noindent{\bf Acknowledgements:} We thank one of the referees for importants comments. P.B thanks E. Sklyanin for a discussion at the early stage of this work, J. Stokman and B. Vlaar for explanations about their works on solutions of boundary $q-KZ$ equations in relation with Koornwinder polynomials, J. Van Diejen and E. Koelink for sending us reprints \cite{D2} and \cite{KDJ}, respectively. P.B thanks  E. Koelink for detailed explanations on the work  \cite{KDJ}, as well as P. Iliev for communications and sending us a reprint of \cite{GerI2}. We thank P. Terwilliger for important comments on the manuscript and sending us the preprint \cite{Terw0}. P.B. is supported by C.N.R.S.

\vspace{1.5cm}

\newpage

\centerline{\bf APPENDIX: Useful formulae for the coefficients}

\vspace{3mm}

$\bullet$ For $k=1,...,N-1$:

\vspace{1mm}

\beqa
\kb(B_k^{0,1}(z))&=&\big(1-\frac{\alpha^2_{N+2-k}}{\alpha_0^2} q^{4\kN_{N-k}+2n_{N+1-k}-2}\big)(1-q^{-2n_{N+1-k}}),\nonumber\\
\kb(B_k^{1,0}(z))&=&\big(1-\frac{\alpha^2_{N+2-k}}{\alpha_0^2} q^{4\kN_{N-k}+2n_{N+1-k}-2}\big)\big(1-\frac{\alpha^2_{N+2-k}}{\alpha^2_{N+1-k}}q^{2n_{N+1-k}}\big),\nonumber\\
\kb(B_k^{0,-1}(z))&=&\big(1-\frac{\alpha^2_{N+2-k}}{\alpha_{N+1-k}^2} q^{2n_{N+1-k}-2}\big)\big(1-\frac{\alpha_0^2}{\alpha^2_{N+1-k}}q^{-4\kN_{N-k}-2n_{N+1-k}+2}\big),\nonumber\\
\kb(B_k^{-1,0}(z))&=&\big(1-q^{-2n_{N+1-k}}\big)\big(1-\frac{\alpha_0^2}{\alpha^2_{N+1-k}}q^{-4\kN_{N-k}-2n_{N+1-k}+2}\big),\nonumber\\
\kb(B_k^{1,-1}(z))&=&\big(1-\frac{\alpha^2_{N+2-k}}{\alpha^2_{N+1-k}}q^{2n_{N+1-k}}\big)\big(1-\frac{\alpha^2_{N+2-k}}{\alpha^2_{N+1-k}}q^{2n_{N+1-k}+2}\big),\nonumber\\
\kb(B_k^{-1,1}(z))&=&\big(1-q^{-2n_{N+1-k}}\big)\big(1-q^{-2n_{N+1-k}+2}\big).\nonumber
\eeqa

\vspace{5mm}

$\bullet$ For $k=0,N$:

\vspace{1mm}

\beqa
\kb(B_0^{0,1}(z))&=&\big(1-\alpha_{N+1}\alpha_{N+2}q^{2\kN_{N}}\big)\big(1-\frac{\alpha_{N+1}\alpha_{N+2}}{\alpha_0^2}q^{2\kN_{N}}\big),\nonumber\\
\kb(B_0^{0,-1}(z))&=&\big(1-\frac{\alpha_{N+2}\alpha_0^2}{\alpha_{N+1}}q^{-2\kN_{N}+2}\big)\big(1-\frac{\alpha_{N+2}}{\alpha_{N+1}}q^{-2\kN_{N}+2}\big),\nonumber\\
\kb(B_N^{1,0}(z))&=&\big(1-\frac{\alpha_2^2}{\alpha_0^2}q^{2n_1-2}\big)\big(1-\frac{\alpha_2^2}{\alpha_1^2}q^{2n_1}\big),\nonumber\\
\kb(B_N^{-1,0}(z))&=&\big(1-q^{-2n_1}\big)\big(1-\frac{\alpha_0^2}{\alpha_1^2}q^{-2n_1+2}\big)
.       \nonumber
\eeqa


\begin{thebibliography}{10}
\bibitem{F1}
Note that the Askey-Wilson algebra first appeared in the context of Leonard pairs under the name 'Leonard algebra'. See \cite[Definition 3.1]{Terw0}.
\bibitem{F2}
Note that the two-variable case was constructed in \cite{KDJ}.
\bibitem{F4}
For triangular, diagonal or special $U_q(sl_2)$ invariant boundary conditions, quotients of the $q-$Onsager algebra have to be considered.
\bibitem{F5}
For the Ising or superintegrable Potts model
  ($q=1$), $h_{-k}=\delta_{k0}$, $h'_0=0$ with Kronecker symbol
  $\delta_{kl}$. For the open XXZ chain, the explicit expressions for
  $h_{-k}$ and $h'_0$ are given in \cite{BK2}.
\bibitem{F6}
The parameters
  $\omega_0,\omega_1,g_+,g_-$ are model-dependent. For the Ising and
  superintegrable chiral Potts model ($q=1$), $\omega_0=1$, $\omega_1$
  is the free magnetic field coupled with the system. For the open XXZ
  chain, $\omega_0,\omega_1,g_+,g_-$ are associated with the boundary
  parameters \cite{BK2}.
\bibitem{F7}
Compared to the notation $E_{q,z_j}$ in \cite{Iliev}, one has $\oE_{z_j}\equiv E_{q^2,z_j}$.
\bibitem{F8}
The $q-$commutator $\big[X,Y\big]_q=qXY-q^{-1}YX$, where $q$ is the deformation parameter, is introduced.
\bibitem{F9}
If $a>1$, $b,c,d$ are real or one is real and the other two are complex conjugates, max$(|b|,|c|,|d|)<1$
and the pairwise products of $ab,ac,ad,bc,bd$ have modulus less than one, then the polynomials satisfy a different orthogonality relation \cite[eq. (3.1.3)]{KS}.
\bibitem{F11}
For applications to quantum integrable systems generated by the $q-$Onsager algebra, scalar products of the Hamiltonian's eigenstates will be expressed in terms of the overlap coefficients.
\bibitem{F12}
By construction, note that the coefficient $c_0=0$.
\bibitem{F13}
Note that a different embedding can be considered, see \cite[Proposition 2.2]{Bas3}.
\bibitem{F14}
We especially thank Paul Terwilliger for pointing out  Proposition 1.24  and other results of \cite{ITaug} to us.
\bibitem{F15}
It is the case for the open XXZ spin chain, see \cite{BK2}.
\bibitem{F16}
In the literature, the notation $r_n(\lambda(x);a,b,c,d)$ with $\lambda(x)=x(x+c+d+1)$ is sometimes used \cite{KS}.
\bibitem{F17}
Recall that zeroes of Askey-Wilson polynomials satisfy Bethe equations associated with a `1'-site open XXZ chain, as pointed out in
\cite{WZ}.
\end{thebibliography}
\end{document}